\documentclass[12pt]{article}
\usepackage{bbm}
\usepackage{mathrsfs}
\usepackage{amsfonts}
\usepackage{graphicx}
\usepackage{amsmath}
\usepackage{amssymb}
\usepackage{dsfont}
\usepackage{float}
\usepackage{natbib}
\usepackage{subcaption}
\usepackage{rotating,color}
\usepackage{xcolor}
\usepackage{booktabs}
\usepackage[colorlinks=true,linkcolor=blue,citecolor=blue,urlcolor=blue]{hyperref}
\usepackage{multirow}
\usepackage{enumitem}
\usepackage{comment}
\usepackage{amsthm}
\usepackage{mathtools}
\usepackage{rotating}
\usepackage{placeins}
\usepackage{bm}

\usepackage{lscape}

\usepackage{amssymb}
\usepackage{threeparttable}

\newtheorem{remark}{Remark}[section]
\newtheorem{theorem}{Theorem}[section]
\newtheorem{assumption}{Assumption}[section]
\newtheorem{lemma}[theorem]{Lemma}
\newtheorem{proposition}[theorem]{Proposition}

\def\bm{\boldsymbol}

\title{Change-Point Detection for Heterogeneous High-Dimensional Functional Time Series}

\author{
Xufei Tang
\and Dan Zhuang\thanks{Corresponding author.}
\and Houlin Zhou\thanks{Co-corresponding author.}
}

\date{\today}

\begin{document}

\maketitle

\begin{abstract}

High-dimensional functional panels consist of temporally ordered curves observed across many subjects and naturally exhibit heterogeneous structural changes. Under sparse subject-level break signals or opposite-signed shifts, traditional mean-aggregated CUSUM procedures may suffer noticeable power loss due to signal attenuation or cancellation induced by cross-sectional averaging. 
We propose a novel Energy--PE statistic, which combines subject-wise squared CUSUM energy aggregation with a generalized power-enhancement component. 
The energy aggregation preserves subject-level evidence under sign-heterogeneous changes, while the power-enhancement component improves sensitivity to sparse weak break signals. Under regularity conditions, we establish the asymptotic behavior of the proposed statistic. We further incorporate a latent group structure and an information-criterion-based clustering algorithm to estimate the unknown group number and membership for heterogeneous break points. Numerical studies and an intraday stock application demonstrate that Energy--PE controls size, improves power under sparse and sign-heterogeneous alternatives, and yields interpretable post-test summaries.

{\it Keywords}: Change-point; Energy aggregation; Power-enhanced testing; Heterogeneous structural breaks; Post-test localization
\end{abstract}

\section{Introduction}
\label{sec:introduction}

Functional time series arise when observations are recorded as curves over a continuum, as in electricity demand profiles, renewable energy generation curves, intraday financial trajectories and environmental monitoring records \citep{Kokoszka12,Yu12,Xu20,Qu21}. In many current applications, curves are observed simultaneously for many subjects, assets or locations, leading to high-dimensional functional panels whose cross-sectional dimension can be comparable to the time span. This setting has motivated recent advances in high-dimensional functional inference, including inference for functional time series and classification of high-dimensional functional data \citep{LiYang23,Xue24}. A central inferential question is whether the mean functions of such curves contain structural breaks over time, since these breaks often correspond to regime shifts, abnormal events or changes in operating patterns.

Methodological work on functional change-point detection has largely focused on testing whether a functional sequence contains a mean break. One line of research reduces each curve to a finite-dimensional representation, often through functional principal components or related projections, and then applies multivariate CUSUM or score-based tests. Representative work includes tests for mean changes \citep{Berkes09}, stability tests for functional autoregressive processes \citep{Hor10}, procedures for dependent functional observations \citep{Aston12}, self-normalized tests for temporally dependent curves \citep{Zhang11} and stationarity tests for functional time series \citep{Hor14}. A complementary line avoids preliminary dimension reduction and constructs test statistics directly from the functional observations. Examples include sequential block bootstrap methods in Hilbert spaces \citep{Sharipov16}, tests for spatiotemporal mean changes \citep{Gromenko17}, full-functional CUSUM procedures \citep{Aue18} and stability tests for functional event observations \citep{Hor22}. These methods have clarified how to detect structural breaks in functional data, but they are mostly developed for a single or low-dimensional collection of functional series.

Beyond global evidence for break existence, a complete change-point analysis also requires localization and structural interpretation. For a single functional sequence, this includes estimating the break location and quantifying its convergence behaviour \citep{Aue09}. For multiple breaks, segmentation-based methods have been developed through dynamic segmentation \citep{Chiou19}, Bayesian wavelet modelling \citep{Li21}, binary segmentation \citep{Rice22}, greedy segmentation \citep{Chen21} and scalable distribution-free procedures \citep{Harris21}; general reviews of offline change-point detection are provided by \cite{Truong20} and \cite{Cho21}. In panel or multi-subject functional data, post-detection inference can be more demanding because different subjects may contain different break locations or no break at all. Related work has considered asynchronous change-point estimation for spatially correlated functional time series \citep{Wang22}, high-dimensional functional CUSUM tests with power enhancement and latent break groups \citep{Li23}, and structural-break detection in high-dimensional functional time-series factor models \citep{XuSuLiuYou26}.

Recent top-journal developments also point to a broader movement toward scalable, high-dimensional and inference-aware change-point methodology. For functional time series, empirical energy-distance methods now allow change-point testing for dependent functional observations and distributional changes beyond mean shifts \citep{BonieceHorvathTrapani25}. In high-dimensional and econometric time series, recent work has developed dimension-agnostic change-point tests that accommodate complex cross-sectional and temporal dependence \citep{GaoWangShao25}, inference for high-dimensional linear-regression changes without exact sparsity \citep{ChoKleyLi25}, and estimating-function approaches for breaks in weak location time-series models \citep{FrancqTrapaniZakoian26}. For offline multiple change-point analysis, narrowest significance pursuit gives finite-sample regions of significant change in linear models \citep{Fryzlewicz24}. In parallel, high-dimensional functional-data methods have advanced rapidly through sparse covariance-function estimation \citep{FangGuoQiao24}, graphical principal component analysis for multivariate functional time series \citep{TanLiangGuanHuang24}, and high-dimensional functional time-series modelling and prediction \citep{ChangFangQiaoYao24}. These contributions are closely aligned with the present goal of retaining functional information while handling high dimensionality, dependence and heterogeneous structural changes.

Despite these developments, several features of high-dimensional functional panels remain difficult to handle within a single framework. Functional curves may be cross-sectionally dependent, subject-specific mean and jump functions may be heterogeneous, and break locations need not be aligned. The signal pattern may also vary across alternatives: some changes are dense but weak, others are sparse but strong, and subject-level jumps may have opposite directions. In such settings, mean-aggregated CUSUM statistics can lose power when opposite-signed changes cancel before detection, whereas purely subject-level screening can overlook broad but weak accumulated evidence. These considerations call for an inferential framework that harnesses subject-level functional information, adapts to sparse strong signals, and connects global testing with interpretable post-test organization of heterogeneous break locations.

Motivated by these considerations, we develop an Energy--PE framework for change-point inference in high-dimensional functional panels with heterogeneous mean changes. Building on the power-enhancement principle of \cite{Fan15} and the high-dimensional functional CUSUM framework of \cite{Li23}, the proposed statistic preserves subject-wise functional CUSUM processes, aggregates their squared \(L^2\) energies, and then adds a generalized enhancement component based on subject-level energy scores. This ordering preserves accumulated functional evidence under heterogeneous jump directions while retaining sensitivity to sparse strong changes. After global rejection, the same subject-level energy scores are used to screen changed subjects, estimate individual break points, recover latent common-break groups and refine group-level locations through pooled estimation.

Overall, this paper develops a unified Energy--PE framework for global testing, subject screening, and break-point organization in high-dimensional functional panels with heterogeneous mean changes. Its main contributions are summarized as follows.
\begin{enumerate}[label=(\arabic*),leftmargin=2.2em,itemsep=0.35em]
\item Motivated by the vulnerability of mean-aggregated CUSUM statistics to heterogeneous and sign-canceling changes, we introduce a novel Energy--PE statistic that redefines how high-dimensional functional break evidence is aggregated. The key obstacle is that cross-sectional averaging can erase subject-level break information before the test is formed, while purely sparse screening can miss broad but weak accumulated evidence. The proposed statistic overcomes this obstacle by first forming subject-wise functional CUSUM processes and pooling their squared \(L^2\) energies, thereby avoiding sign cancellation and accumulating heterogeneous functional evidence. It then adds a generalized smooth enhancement function \(g\) of the subject-level energy scores to amplify sparse pronounced signals. This two-component construction yields a test that is sensitive to dense, sparse and sign-heterogeneous alternatives while preserving functional information before cross-sectional pooling.

\item We provide a unified theoretical foundation for the proposed procedure, covering global testing, subject screening, break localization, latent group recovery and pooled group-level refinement. For the global statistic, the theory establishes its asymptotic behaviour under regularity conditions, shows that the generalized enhancement function \(g\) can enlarge sensitivity to sparse alternatives without changing the leading null behaviour, and characterizes the power contribution of the energy component under accumulated-energy alternatives. For post-test inference, the results give conditions under which changed subjects, heterogeneous break locations and latent common-break groups can be recovered consistently. This theory therefore links detection and structural estimation in a single high-dimensional functional framework rather than treating them as separate tasks.

\item We demonstrate the practical advantages of the method through extensive simulations and an empirical financial application. The simulation study in Section~\ref{sec:simulation} compares Energy--PE with benchmark procedures under different functional data-generating mechanisms, sample sizes, sparsity regimes, signal strengths and jump-direction patterns, showing stable size performance and clear power gains in the challenging sparse and sign-heterogeneous settings. It also evaluates post-test subject identification, break-location estimation, the effect of different choices of \(g\), and latent common-break recovery. The stock-data analysis in Section~\ref{sec:empirical} shows that the proposed workflow can detect a sparse subset of structurally changed functional time series and summarize their heterogeneous break dates through interpretable latent groups.
\end{enumerate}

The remainder of this paper is organized as follows. Section~\ref{sec:model} introduces the high-dimensional functional panel model and the proposed Energy--PE statistic, together with a comparison to the mean-aggregated PE--CUSUM benchmark. Section~\ref{sec:theory} establishes the theoretical properties of the proposed procedure and develops the post-test estimation steps for identifying changed subjects, recovering latent break groups, and refining pooled break locations. Section~\ref{sec:simulation} presents simulation evidence, Section~\ref{sec:empirical} provides an empirical application, and Section~\ref{sec:conclusion} concludes. Technical details and proofs are collected in Appendices~\ref{app:A}--\ref{app:additional-experiment}.

We use the following notation throughout. For a square-integrable function \(f\) on the compact domain \(C\), \(\|f\|_{L^2(C)}^2=\int_C f^2(u)\,du\). The indicator function is denoted by \(I(\cdot)\). For real numbers \(a\) and \(b\), \(a\vee b=\max(a,b)\), and \(\lfloor a\rfloor\) denotes the integer part of \(a\). The symbols \(O_P(\cdot)\), \(o_P(\cdot)\), and ``w.p.a.1'' have their usual asymptotic meanings. Unless otherwise stated, limits are taken jointly as \(N,T\to\infty\).

\section{Model and Test Statistic}
\label{sec:model}

This section introduces the model and constructs the proposed Energy--PE statistic. We first formulate a high-dimensional functional panel model that allows subject-specific mean functions, jump functions, and break locations to vary across subjects. We then build the statistic in three steps. The subject-wise Energy--CUSUM component aggregates squared functional CUSUM energies and is therefore robust to sign heterogeneity; the sparse-signal enhancement component amplifies unusually large subject-level energies; and the final Energy--PE statistic combines these two sources of evidence in an additive form. The section closes with a comparison to mean-aggregated PE--CUSUM, highlighting why the proposed ordering of squaring and cross-sectional aggregation avoids sign cancellation.

Suppose that we observe a panel of functional time series
${\mathbf X}_t=(X_{1t},\cdots,X_{Nt})^{\intercal}$,
$t=1,\cdots,T$, where
$X_{it}=\left(X_{it}(u): u\in C\right)$ and $C$ is a compact functional domain.
For the $i$th subject, the observations are generated from the mean-shift model
\begin{equation}\label{eq2.1}
X_{it}=\mu_i+\delta_i I\left(t>\tau_i\right)+\varepsilon_{it}.
\end{equation}
Here $\mu_i=(\mu_i(u): u\in C)$ is the pre-break mean function, $\delta_i=(\delta_i(u): u\in C)$ is the jump function, $\tau_i$ is the break point, and $\varepsilon_{it}=(\varepsilon_{it}(u): u\in C)$ is a stationary functional error process over time. The mean functions, jump functions, and break locations are all allowed to vary across subjects, which accommodates both cross-sectional heterogeneity in signal shape and heterogeneity in break timing. The global testing problem is
\begin{equation}\label{eq2.2}
H_0:\ \delta_i=0,\ i=1,\cdots,N,\ \ \ \ {\rm versus}\ \ \ \ H_A:\ \delta_i\neq 0\ {\rm for\ some}\ i.
\end{equation}
Under $H_0$, model~\eqref{eq2.1} reduces to the stationary model $X_{it}=\mu_i+\varepsilon_{it}$. The definitions below use the same subject-wise CUSUM energy scores for three purposes: constructing the global test statistic, screening changed subjects after rejection, and estimating individual break locations.

\subsection{Subject-Wise Energy--CUSUM Component}
\label{subsec:energy-cusum}

For each subject $i$, define the functional CUSUM process
\begin{equation}
C_{i,T}(x,u)
=
\frac{1}{\sqrt{T}}
\left[
\sum_{t=1}^{\lfloor Tx\rfloor} X_{it}(u)
-
\frac{\lfloor Tx\rfloor}{T}\sum_{t=1}^{T} X_{it}(u)
\right],
\qquad 0\le x\le 1.
\label{eq:cusum_process}
\end{equation}
The \(L^2\) energy of \(C_{i,T}\) measures the evidence for a mean change in subject \(i\) at the candidate fraction \(x\). We aggregate these energies across subjects by
\begin{equation}
E_T(x)
=
\frac{1}{N}\sum_{i=1}^{N}\int_C C_{i,T}^2(x,u)\,du,
\qquad 0\le x\le 1,
\label{eq:energy_curve}
\end{equation}
and define the Energy--CUSUM component
\begin{equation}
Z_T^{\mathrm{Energy}}
=
\sup_{0\le x\le 1} E_T(x).
\label{eq:energy_stat}
\end{equation}

The order of operations in \eqref{eq:energy_curve} is essential: each subject-specific CUSUM process is squared before the cross-sectional average is taken. Consequently, changes with opposite directions reinforce each other through their squared energies rather than canceling through a linear average. This feature makes \(Z_T^{\mathrm{Energy}}\) well suited to dense and moderately dense alternatives, including symmetric sign-canceling configurations.

\subsection{Sparse-Signal Enhancement}
\label{subsec:sparse-enhancement}

The energy aggregation above is most effective when signals accumulate across many subjects. To retain sensitivity when only a small number of subjects have pronounced changes, we also use the subject-level maximal energy
\begin{equation}
Y_{iT}
=
\sup_{0\le x\le 1}\int_C C_{i,T}^2(x,u)\,du,
\label{eq:subject_energy}
\end{equation}
where $C_{i,T}$ is defined in \eqref{eq:cusum_process}. For a threshold $\xi_{NT}$, scale parameter $\tau_{NT}>0$, normalization factor $r_{NT}>0$, and measurable nondecreasing score function $g:\mathbb R\to[0,\infty)$, define the sparse-signal enhancement term
\begin{equation}
\mathcal Z_{NT}^{g}
=
r_{NT}\sum_{i=1}^{N}
g\!\left(\frac{Y_{iT}-\xi_{NT}}{\tau_{NT}}\right).
\label{eq:enhancement_general}
\end{equation}

This construction separates the three roles of the enhancement step: \(\xi_{NT}\) determines which subject-level energies are unusually large, \(\tau_{NT}\) controls the scale of the exceedance, and \(g\) determines how exceedances are counted or weighted. Writing
\begin{equation}
S_{iT}=\frac{Y_{iT}-\xi_{NT}}{\tau_{NT}},
\label{eq:standardized_energy}
\end{equation}
the enhancement term is the normalized total score
\[
\mathcal Z_{NT}^{g}
=
r_{NT}\sum_{i=1}^{N}g(S_{iT}).
\]
Several useful choices of \(g\) are:
\begin{itemize}[leftmargin=3.0em,labelwidth=2.2em,labelsep=0.4em,itemsep=0.35em]
\item[(1)] \textbf{Hard-threshold score:}
\[
g(z)=\mathbf 1\{z>0\}.
\]

\item[(2)] \textbf{Excess-sum score:}
\[
g(z)=z_+=\max\{z,0\}.
\]

\item[(3)] \textbf{Smooth thresholding scores:}
\[
g(z)=\log(1+e^z)
\qquad \text{or} \qquad
g(z)=(1+e^{-z})^{-1}.
\]
\end{itemize}
Thus hard-threshold, excess-sum, and smooth thresholding rules are all covered by the same total-score construction. The theoretical results below are stated for \(\mathcal Z_{NT}^{g}\); the hard-threshold rule is recovered by taking \(g(z)=\mathbf 1\{z>0\}\).

\begin{remark}[Choice of the enhancement score]
The indicator score $g(z)=\mathbf 1\{z>0\}$ records only whether a subject-specific CUSUM energy exceeds the threshold. With $r_{NT}=\sqrt{N\vee T}$ and $\tau_{NT}=1$, it gives the hard-threshold enhancement
\[
\mathcal Z_{NT}^{\diamond}
=
\sqrt{N\vee T}\sum_{i=1}^{N}\mathbf 1\!\left(Y_{iT}>\xi_{NT}\right).
\]
This choice is simple and often convenient for proving null negligibility, but it discards the size of the exceedance. The positive-part score $g(z)=z_+$ instead yields an excess-sum statistic,
\[
\mathcal Z_{NT}^{g}
=
r_{NT}\sum_{i=1}^{N}
\left(\frac{Y_{iT}-\xi_{NT}}{\tau_{NT}}\right)_+,
\]
or equivalently $\widetilde r_{NT}\sum_{i=1}^{N}(Y_{iT}-\xi_{NT})_+$ after absorbing $\tau_{NT}$ into the normalization. Smooth scores such as $\log(1+e^z)$ or $(1+e^{-z})^{-1}$ interpolate continuously around the threshold and may be numerically more stable when some $Y_{iT}$'s are close to $\xi_{NT}$, although their null analysis requires tail-expectation rather than pure exceedance-probability control.
\end{remark}

\subsection{Energy--PE Statistic}
\label{subsec:energy-pe}

The proposed global test statistic combines the subject-wise energy aggregation and the sparse-signal enhancement in an additive form. This choice reflects the two ways in which evidence can appear in a high-dimensional functional panel: a broad accumulation of moderate subject-level energies, or a small number of unusually large subject-level energies. We define
\begin{equation}
\widehat Z_{NT}^{\mathrm{EPE}}
=
Z_T^{\mathrm{Energy}}+\mathcal Z_{NT}^{g}.
\label{eq:new_stat}
\end{equation}
The Energy--CUSUM term \(Z_T^{\mathrm{Energy}}\) provides robustness to heterogeneous signs and gains power when many subjects contribute moderate evidence. The enhancement term \(\mathcal Z_{NT}^{g}\) contributes little under the null but can dominate when a small set of subjects has sufficiently large CUSUM energies. Hence \(\widehat Z_{NT}^{\mathrm{EPE}}\) is designed to adapt to both dense sign-heterogeneous regimes and sparse strong-signal regimes without requiring the sparsity level to be specified in advance.

From a high-dimensional $\ell^2$ inference perspective, the null distribution of the energy term is characterized through a Gaussian approximation to the full panel CUSUM process rather than through the Brownian-bridge limit of a linearly pooled CUSUM. The enhancement term is calibrated to be asymptotically negligible under the null, so the first-order null behavior of \(\widehat Z_{NT}^{\mathrm{EPE}}\) remains governed by the calibrated Energy--CUSUM component, while its power can be enlarged under sparse alternatives.

\begin{remark}[Difference from PE--CUSUM]
The benchmark PE--CUSUM method first forms a cross-sectional average functional CUSUM and then adds a subject-level hard-threshold enhancement term. By contrast, the proposed Energy--PE statistic first keeps the subject-specific CUSUM processes, squares their \(L^2\) energies, and only then aggregates across subjects. This order of operations is important: if half of the changed subjects have jump function \(\delta(\cdot)\) and the other half have \(-\delta(\cdot)\), the mean jump can vanish even though the total squared jump energy is positive. Figure~\ref{fig:counterexample} illustrates this sign-cancellation pattern. The detailed construction and proof are deferred to Appendix~\ref{app:A}.
\end{remark}

\begin{figure}[!htbp]
    \centering
    \includegraphics[width=0.85\textwidth]{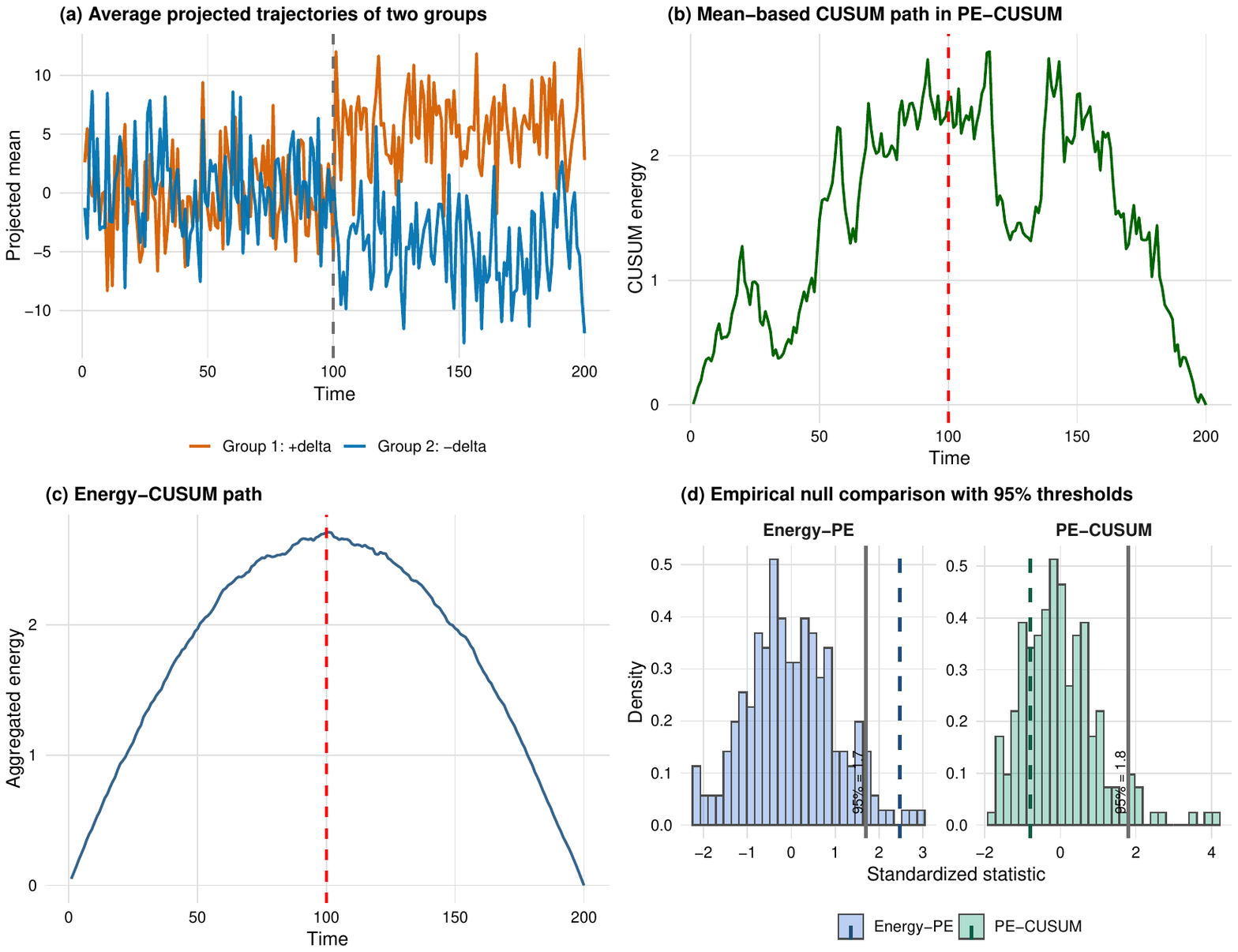}
    \caption{Subject-wise energy aggregation preserves opposite-signed break signals. The sign-cancellation counterexample is generated with $T=200$, $N=400$, and $\delta=0.1$. (a) Average projected trajectories for two subject groups with opposite-sign breaks. (b) The mean-based CUSUM path used by PE--CUSUM, which is weakened by cross-sectional cancellation. (c) The proposed Energy--CUSUM path, which aggregates squared subject-specific break energies and attains a clear maximum near the true break point. (d) Empirical null comparison based on standardized statistics. In each panel, the dashed line denotes the observed statistic and the solid vertical line denotes the empirical 95\% threshold under $H_0$.}
    \label{fig:counterexample}
\end{figure}
\FloatBarrier

\section{Theoretical Properties and Post-Test Estimation}
\label{sec:theory}

This section develops the asymptotic foundation for the proposed Energy--PE procedure. We first state the regularity and signal conditions needed for the global test. We then establish the asymptotic properties of the statistic, including null calibration, power under accumulated energy alternatives, and the contribution of the generalized enhancement term. Finally, we study the post-test estimation problem: the subject-wise energy scores are used to screen changed subjects, localize their preliminary break points, recover latent common-break groups, and refine group-level locations by pooled estimation.

\subsection{Assumptions for the Global Test}
\label{subsec:global-assumptions}

We first state the conditions used to analyze the global statistic. The first assumption controls the stochastic behavior of the subject-wise CUSUM energy process.

\begin{assumption}[Regularity for the Energy--CUSUM process]
\label{ass:regularity}
\normalfont
\leavevmode\par\vspace{0.25em}
\begin{enumerate}[label=(\arabic*),leftmargin=2.2em,itemsep=0.35em]

\item \textit{Dependence.} For each $i=1,\ldots,N$, the error process admits the linear representation
\[
\varepsilon_{it}
=
\sum_{j=0}^{\infty}A_{ij}\eta_{i,t-j},
\]
where $\eta_t=(\eta_{1t},\ldots,\eta_{Nt})$ are independent over $t$, centered $L^2(C)^N$-valued random elements with uniformly bounded fourth moments. The bounded linear operators $A_{ij}$ satisfy
\[
\sum_{j=0}^{\infty} j \max_{1\le i\le N}\|A_{ij}\|_{\mathcal O}<\infty.
\]

\item Define the product Hilbert space
\[
\mathcal H_N=\{h=(h_1,\ldots,h_N):h_i\in L^2(C)\},
\qquad
\|h\|_{\mathcal H_N}^2=\frac1N\sum_{i=1}^N\|h_i\|_{L^2(C)}^2,
\]
and write $\mathcal C_{N,T}(x)=(C_{1,T}(x),\ldots,C_{N,T}(x))$. Under $H_0$, the second moment of the energy process is uniformly bounded and its cross-sectional average is stable:
\[
\sup_{0\le x\le1}
\left|
\frac{1}{N}\sum_{i=1}^N
\left\{
\int_C C_{i,T}^2(x,u)\,du
-
E\int_C C_{i,T}^2(x,u)\,du
\right\}
\right|
=o_P(1),
\]
and
\[
\sup_{0\le x\le1}\frac{1}{N}\sum_{i=1}^N
E\int_C C_{i,T}^2(x,u)\,du
<\infty.
\]

\item There exists $\kappa>0$ such that
\[
N=O(T^\kappa).
\]
\end{enumerate}
\end{assumption}

\begin{remark}[Regularity conditions]
Assumption \ref{ass:regularity} is stated for the statistic actually used in \eqref{eq:energy_stat}. Because \(Z_T^{\mathrm{Energy}}\) averages subject-wise squared CUSUM energies, its null behavior differs from that of a squared pooled CUSUM process. The product-space normalization in Assumption \ref{ass:regularity}(2) therefore provides the scale on which the Gaussian approximation in Theorem~\ref{thm:3.1} is formulated.
\end{remark}

The next condition describes alternatives under which subject-level CUSUM energies accumulate across the panel. It includes dense alternatives and sign-canceling alternatives that are invisible after linear cross-sectional averaging.

\begin{assumption}[Dense and sign-canceling alternatives]
\label{ass:dense}
\normalfont
For dense or sign-canceling alternatives, there exists a subset
$I_N\subset\{1,\ldots,N\}$ and constants $r_0\in(0,1)$, $\eta\in(0,1/2)$ such that
\begin{equation}
\max_{i\in I_N}\left|\frac{\tau_i}{T}-r_0\right|\to 0,
\qquad
\eta\le r_0\le 1-\eta,
\label{eq:3.B4}
\end{equation}
and
\begin{equation}
\frac{T}{N}\sum_{i\in I_N}\|\delta_i\|_{L^2(C)}^2\to\infty.
\label{eq:3.B5}
\end{equation}
\end{assumption}

\begin{remark}[Dense and sign-canceling alternatives]
Assumption \ref{ass:dense} is the signal accumulation condition for the Energy--CUSUM component. The alignment condition \eqref{eq:3.B4} ensures that the deterministic CUSUM signals are evaluated on a common time scale, while \eqref{eq:3.B5} requires the accumulated squared jump energy to dominate stochastic fluctuation.
\end{remark}

Finally, the enhancement term requires one condition under the null and another under sparse strong alternatives.

\begin{assumption}[Calibration and sparse-power conditions for the enhancement term]
\normalfont
\phantomsection\label{ass:enhancement}
Here \(S_{iT}=(Y_{iT}-\xi_{NT})/\tau_{NT}\), and \(g:\mathbb R\to[0,\infty)\) is the nonnegative nondecreasing score function used in \eqref{eq:enhancement_general}.
\leavevmode\par\vspace{0.25em}
\begin{enumerate}[label=(\arabic*),leftmargin=2.2em,itemsep=0.35em]
\item \textit{Null control.} The threshold $\xi_{NT}$, scale parameter $\tau_{NT}$, normalization factor $r_{NT}$, and transformation $g$ satisfy
\begin{equation}
\begin{gathered}
\xi_{NT}\to\infty,\\
Nr_{NT}\sup_i E\{g(S_{iT})\mid H_0\}\to0.
\end{gathered}
\label{eq:3.B6}
\end{equation}

\item For sparse strong alternatives, there exists a subset $J_N\subset\{1,\ldots,N\}$ and a positive sequence $b_{NT}$ such that
\begin{equation}
\inf_{i\in J_N}
P\!\left(
Y_{iT}\ge \xi_{NT}+\tau_{NT}b_{NT}
\right)\to1
\label{eq:3.B7}
\end{equation}
and
\begin{equation}
r_{NT}|J_N|\,g(b_{NT})\to\infty.
\label{eq:3.B8}
\end{equation}
\end{enumerate}
\end{assumption}

\begin{remark}[Enhancement calibration]
Assumption \ref{ass:enhancement}(1) is the null negligibility condition. For the indicator rule
\[
g(z)=\mathbf 1\{z>0\},
\]
it reduces to the exceedance-probability control
\[
Nr_{NT}\sup_{1\le i\le N}P(Y_{iT}>\xi_{NT}\mid H_0)\to0.
\]
For \(g(z)=z_+\), the same condition becomes the tail-expectation bound
\[
Nr_{NT}\sup_{1\le i\le N}
E\!\left[
\left(\frac{Y_{iT}-\xi_{NT}}{\tau_{NT}}\right)_+
\middle|H_0
\right]\to0.
\]
Assumption \ref{ass:enhancement}(2) is the corresponding sparse-signal condition. It requires affected subjects to exceed the threshold by a non-negligible amount and requires their aggregate transformed exceedance to diverge.
\end{remark}

\subsection{Asymptotic Properties of the Statistic}
\label{subsec:global-validity-power}

We next establish validity of the global test. The first result gives the null calibration of the Energy--CUSUM component.

\begin{theorem}[Null calibration of Energy--CUSUM]
\label{thm:3.1}
Suppose that $H_0$ and Assumption \ref{ass:regularity} hold. Let
\[
G_N(x)=(G_{1,N}(x),\ldots,G_{N,N}(x)),\qquad 0\le x\le1,
\]
be a centered Gaussian bridge in $\mathcal H_N$ with the same long-run covariance structure as the panel CUSUM process under $H_0$, and define
\begin{equation}
Z_{NT}^{G}
=
\sup_{0\le x\le1}
\frac1N\sum_{i=1}^N\|G_{i,N}(x)\|_{L^2(C)}^2,
\label{eq:3.4}
\end{equation}
where the covariance is understood in the product-space norm of Assumption \ref{ass:regularity}(2). Assume, in addition, that the Gaussian approximation
\begin{equation}
\sup_{z\in\mathbb R}
\left|
P_{H_0}\!\left(Z_T^{\mathrm{Energy}}\le z\right)
-
P\!\left(Z_{NT}^{G}\le z\right)
\right|\to0
\label{eq:3.GA}
\end{equation}
holds. Let $Z_{NT}^{*,\mathrm{Energy}}$ denote the multiplier bootstrap analogue of $Z_T^{\mathrm{Energy}}$, and let $c_{NT,\alpha}^{*}$ be its conditional upper $\alpha$-level critical value. If the bootstrap consistently estimates the Gaussian law,
\begin{equation}
\sup_{z\in\mathbb R}
\left|
P^*\!\left(Z_{NT}^{*,\mathrm{Energy}}\le z\right)
-
P\!\left(Z_{NT}^{G}\le z\right)
\right|\xrightarrow{P}0,
\label{eq:3.bootstrap}
\end{equation}
the bootstrap critical values are tight, $c_{NT,\alpha}^{*}=O_P(1)$, and the Gaussian approximation is asymptotically anti-concentrated in the sense that, for every deterministic sequence \(\eta_{NT}\downarrow0\),
\[
\sup_{z\in\mathbb R}
P\!\left(|Z_{NT}^{G}-z|\le \eta_{NT}\right)\to0,
\]
then
\[
\limsup_{N,T\to\infty}
P_{H_0}\!\left(Z_T^{\mathrm{Energy}}>c_{NT,\alpha}^{*}\right)
\le \alpha.
\]
\end{theorem}

The next theorem records the power contribution of the energy component. It applies whenever the squared subject-wise jump energy accumulates across the panel, even if the signed average jump is close to zero.

\begin{theorem}[Power under accumulated energy alternatives]
\label{thm:3.2}
Suppose that $H_A$ and Assumptions \ref{ass:regularity} and \ref{ass:dense} hold. Then, as $N,T\to\infty$ jointly,
\begin{equation}
Z_T^{\mathrm{Energy}}\xrightarrow{P}\infty.
\label{eq:3.5}
\end{equation}
In particular, the above consistency continues to hold under symmetric sign-canceling alternatives for which the linearly aggregated mean shift vanishes but \eqref{eq:3.B5} remains satisfied.
\end{theorem}

Theorem~\ref{thm:3.2} explains the main departure from mean-aggregated PE--CUSUM procedures such as \cite{Li23}: quadratic subject-wise aggregation preserves energy in sign-canceling alternatives that would be removed by linear averaging. The enhancement term plays a different role, contributing primarily under sparse strong alternatives. The two terms are therefore calibrated so that the enhancement is asymptotically negligible under \(H_0\), but can enlarge the power region under alternatives not covered by energy accumulation alone.

\begin{theorem}[Null behavior of Energy--PE]
\label{thm:3.3}
Suppose that the conditions of Theorem \ref{thm:3.1} and Assumption \ref{ass:enhancement}(1) hold. Then, as $N,T\to\infty$ jointly,
\begin{equation}
\mathcal Z_{NT}^{g}=o_P(1),
\qquad
\widehat Z_{NT}^{\mathrm{EPE}}
=
Z_T^{\mathrm{Energy}}+o_P(1).
\label{eq:3.7}
\end{equation}
Consequently, the test that rejects $H_0$ when $\widehat Z_{NT}^{\mathrm{EPE}}>c_{NT,\alpha}^{*}$ has asymptotic size at most $\alpha$.
\end{theorem}

The following theorem combines the two power mechanisms. It makes explicit that the proposed statistic is not tied to a single signal geometry: either accumulated energy or sufficiently strong sparse exceedances are enough for consistency.

\begin{theorem}[Consistency of Energy--PE]
\label{thm:3.4}
Let $c_{NT,\alpha}^{*}$ be the bootstrap critical value in Theorem \ref{thm:3.1} and suppose that $c_{NT,\alpha}^{*}=O_P(1)$. Suppose either
\begin{itemize}
\item[\textup{(1)}] $H_A$ and Assumptions \ref{ass:regularity} and \ref{ass:dense} hold; or
\item[\textup{(2)}] the sparse-power condition in Assumption \ref{ass:enhancement}(2) holds.
\end{itemize}
Then, as $N,T\to\infty$ jointly,
\begin{equation}
P\!\left(\widehat Z_{NT}^{\mathrm{EPE}}>c_{NT,\alpha}^{*}\right)\to1.
\label{eq:3.8}
\end{equation}
\end{theorem}

The two components in \(\widehat Z_{NT}^{\mathrm{EPE}}\) are designed for complementary alternatives. The Energy--CUSUM component is most effective when subject-wise squared jump energies accumulate, whereas the enhancement component is most effective when a relatively small set of subjects produces large standardized energy scores.

\begin{table}[htbp]
\centering
\caption{Dominant components of the Energy--PE statistic across signal regimes.}
\label{tab:dominance}
\begin{threeparttable}
\setlength{\tabcolsep}{6pt}
\renewcommand{\arraystretch}{1.15}
\begin{tabular}{lccc}
\toprule
\textbf{Regime} 
& \(\bm{Z}_T^{\mathbf{Energy}}\) 
& \(\bm{\mathcal Z}_{\mathbf{NT}}^{\mathbf{g}}\) 
& \textbf{Dominant Component} \\
\midrule
\(H_0\) 
& \(O_P(1)\) 
& \(o_P(1)\) 
& \(Z_T^{\mathrm{Energy}}\) \\

\(H_1\) (Dense Weak) 
& \(\xrightarrow{P}\infty\) 
& \(o_P(1)\) or \(O_P(1)\) 
& \(Z_T^{\mathrm{Energy}}\) \\

\(H_1\) (Sparse Strong)  
& \(O_P(1)\) 
& \(\xrightarrow{P}\infty\) 
& \(\mathcal Z_{NT}^{g}\) \\

\(H_1\) (Dense Strong) 
& \(\xrightarrow{P}\infty\) 
& possibly \(\xrightarrow{P}\infty\) 
& both \\
\bottomrule
\end{tabular}
\begin{tablenotes}[flushleft]
\footnotesize
\item[$\dagger$] The table summarizes which component determines the first-order power behavior in each regime. Under \(H_0\), the enhancement component is calibrated to be asymptotically negligible.
\end{tablenotes}
\end{threeparttable}
\end{table}

Table~\ref{tab:dominance} summarizes the dominant components of the proposed Energy--PE test statistic under different regimes. Under the null hypothesis, the PE component is asymptotically negligible, and the rejection rule is governed by the calibrated Energy--CUSUM critical value. For dense or symmetric alternatives, the Energy--CUSUM component diverges and dominates the test statistic. In contrast, under sparse but strong alternatives, the Energy--CUSUM component need not diverge, whereas the PE component diverges and drives the power of the test.

The global statistic and the subject-wise CUSUM energies play different roles. The former is used for testing \(H_0\); the latter provide the individual evidence used after rejection. We next use the same \(Y_{iT}\)'s to screen changed subjects, estimate their break locations, and organize heterogeneous locations into latent common groups.

\subsection{Estimation of Break Points and Latent Structure}
\label{subsec:break-latent-estimation}

After the global null has been rejected, the next inferential task is to determine which subjects changed and where their individual breaks occurred. Define
\[
\mathcal C_1:=\{1\le i\le N:\delta_i\not\equiv0\},
\qquad
\mathcal C_0:=\{1\le i\le N:\delta_i\equiv0\}.
\]
The subject-wise maximal energies \(Y_{iT}\) provide a direct screening statistic. We estimate the changed and unchanged sets by
\[
\widehat{\mathcal C}_1
=
\left\{
1\le i\le N:
Y_{iT}>\xi_{NT}
\right\},
\qquad
\widehat{\mathcal C}_0
=
\left\{
1\le i\le N:
Y_{iT}\le\xi_{NT}
\right\},
\]
where
\[
Y_{iT}
=
\sup_{0\le x\le 1}\int_C C_{i,T}^2(x,u)\,du.
\]
The same symbol $\xi_{NT}$ is used for notational economy; in applications the global enhancement threshold and the screening threshold may be chosen separately, provided that their respective null and signal conditions are satisfied.

For each $i\in\widehat{\mathcal C}_1$, define the preliminary break-point estimator
\begin{equation}
\widehat\tau_i
=
\arg\max_{1\le t\le T}
\int_C C_{i,T}^2\!\left(\frac{t}{T},u\right)\,du.
\label{eq:3.9}
\end{equation}

\begin{assumption}[Screening separation and localization signal]
\label{ass:screening}
\normalfont
Let
\[
Y_{iT}^{(0)}
=
\sup_{0\le x\le1}\int_C D_{i,T}^2(x,u)\,du
\]
denote the deterministic signal energy induced by the mean shift, where \(D_{i,T}\) is the CUSUM transform in \eqref{eq:cusum_process} applied to the signal sequence \(\delta_i\mathbf 1\{t>\tau_i\}\). Suppose that:
\leavevmode\par\vspace{0.25em}
\begin{enumerate}[label=(\arabic*),leftmargin=2.2em,itemsep=0.35em]
\item The false-positive probability is uniformly controlled:
\begin{equation}
N\sup_{i\in\mathcal C_0}P(Y_{iT}>\xi_{NT})\to0,
\label{eq:3.9a}
\end{equation}

\item The changed subjects are separated from the screening threshold:
\begin{equation}
\min_{i\in\mathcal C_1}
\frac{T\omega_{Ti}^2\|\delta_i\|_{L^2(C)}^2}{\xi_{NT}}
\to\infty,
\qquad
\omega_{Ti}:=\left(\frac{\tau_i}{T}\right)\wedge\left(1-\frac{\tau_i}{T}\right).
\label{eq:3.10}
\end{equation}

\item The observed subject-wise energies approximate their deterministic signal energies uniformly over the changed subjects:
\begin{equation}
\max_{i\in\mathcal C_1}
\frac{|Y_{iT}-Y_{iT}^{(0)}|}
{T\omega_{Ti}^2\|\delta_i\|_{L^2(C)}^2}
=o_P(1).
\label{eq:3.10b}
\end{equation}

\item For localization, there exists a constant $c_\delta>0$ such that
\begin{equation}
\min_{i\in\mathcal C_1}\omega_{Ti}\|\delta_i\|_{L^2(C)}\ge c_\delta,
\label{eq:3.11}
\end{equation}
then the break magnitudes are uniformly separated from zero.
\end{enumerate}
\end{assumption}

Assumption \ref{ass:screening} separates changed and unchanged subjects. Condition \eqref{eq:3.9a} controls false positives uniformly over \(\mathcal C_0\), conditions \eqref{eq:3.10}--\eqref{eq:3.10b} ensure that each truly changed subject crosses the screening threshold, and condition \eqref{eq:3.11} is a strengthened signal condition used only for the localization rate.

\begin{assumption}[Localization maximal inequality]
\label{ass:localmax}
\normalfont
For the subject-wise squared CUSUM criterion,
\[
\max_{1\le i\le N}\sup_{1\le t\le T}
\left|
\int_C C_{i,T}^2(t/T,u)\,du
-
E\!\left[\int_C C_{i,T}^2(t/T,u)\,du\right]
\right|
=
O_P\!\left(\log(N\vee T)\right).
\]
\end{assumption}

\begin{remark}[Localization condition]
Assumption \ref{ass:localmax} is used only for the refined rate of the preliminary break-point estimators. It is an exponential-type maximal inequality over the $N\times T$ subject-time grid and is implied, for example, by sub-Gaussian innovations together with sufficiently fast temporal dependence decay. The weaker screening consistency in \eqref{eq:3.12} does not require this strengthened maximal control.
\end{remark}

\begin{theorem}[Screening and preliminary localization]
\label{thm:3.5}
Suppose that Assumption \ref{ass:regularity} and Assumption \ref{ass:screening}(1)--(3) hold. Then, as $N,T\to\infty$ jointly,
\begin{equation}
P\!\left(
\widehat{\mathcal C}_0=\mathcal C_0,\ 
\widehat{\mathcal C}_1=\mathcal C_1
\right)\to1.
\label{eq:3.12}
\end{equation}
Moreover, if Assumption \ref{ass:screening}(4) and Assumption \ref{ass:localmax} hold, then
\begin{equation}
\max_{i\in\mathcal C_1}
|\widehat\tau_i-\tau_i|
=
o_P\!\left([\log(N\vee T)]^{1+\zeta}\right)
\label{eq:3.13}
\end{equation}
for any arbitrarily small $\zeta>0$.
\end{theorem}

Theorem \ref{thm:3.5} shows that the subject-wise energy scores support both consistent screening and preliminary localization. These estimates provide the inputs for the latent grouping step below.

We next organize heterogeneous individual break locations into a small number of latent common break points. Suppose that the set of changed subjects $\mathcal C_1$ admits a latent partition
\[
\mathcal C_1
=
\bigcup_{k=1}^{K_0}\mathcal C(b_k),
\qquad
\mathcal C(b_{k_1})\cap\mathcal C(b_{k_2})=\varnothing
\ \text{for } k_1\neq k_2,
\]
such that
\begin{equation}
\tau_i=b_k,
\qquad
\forall\, i\in\mathcal C(b_k),
\qquad
1\le k\le K_0,
\label{eq:3.14}
\end{equation}
where
\[
1\le b_1<\cdots<b_{K_0}\le T
\]
are $K_0$ distinct break points. Both the group membership and the number \(K_0\) are unknown.

For $i\in\widehat{\mathcal C}_1$, let $\widehat\tau_i$ be defined in \eqref{eq:3.9}. Sort these preliminary estimates from the smallest to the largest:
\[
\widehat\tau_{(1)}\le\cdots\le \widehat\tau_{(n)},
\qquad
n=|\widehat{\mathcal C}_1|.
\]
Define adjacent gaps
\[
\Delta_j(\widehat\tau)=\widehat\tau_{(j+1)}-\widehat\tau_{(j)},
\qquad
j=1,\ldots,n-1.
\]

For a candidate number of groups \(K\), let
$1\le \widehat j_{K,1}<\cdots<\widehat j_{K,K-1}\le n-1$
be the ordered indices of the $(K-1)$ largest gaps in
$\{\Delta_j(\widehat\tau)\}_{j=1}^{n-1}$. Set
$\widehat j_{K,0}=0$ and $\widehat j_{K,K}=n$. The $k$-th estimated cluster is defined by the sorted-index block
\begin{equation}
\widehat{\mathcal C}(k\mid K)
=
\left\{
i\in\widehat{\mathcal C}_1:
\widehat\tau_i=\widehat\tau_{(\ell)}
\text{ for some }
\widehat j_{K,k-1}<\ell\le \widehat j_{K,k}
\right\},
\qquad
k=1,\ldots,K,
\label{eq:3.15}
\end{equation}
with ties broken by the subject index. This index-based definition avoids empty clusters and boundary ambiguities.

For each $k=1,\ldots,K$, define the clusterwise averaged break-point estimator
\begin{equation}
\widehat b_{k\mid K}
=
\frac{1}{|\widehat{\mathcal C}(k\mid K)|}
\sum_{i\in\widehat{\mathcal C}(k\mid K)}
\widehat\tau_i.
\label{eq:3.16}
\end{equation}

For each $i\in \widehat{\mathcal C}(k\mid K)$, define
\[
\widehat\nu_{it,k\mid K}
=
\widehat\mu_{i,k\mid K}
+
\widehat\delta_{i,k\mid K}\mathbf 1\{t>\widehat b_{k\mid K}\},
\]
where
\[
\widehat\mu_{i,k\mid K}
=
\frac{1}{\widehat b_{k\mid K}}
\sum_{t=1}^{\widehat b_{k\mid K}}X_{it},
\qquad
\widehat\delta_{i,k\mid K}
=
\frac{1}{T-\widehat b_{k\mid K}}
\sum_{t=\widehat b_{k\mid K}+1}^{T}X_{it}
-
\frac{1}{\widehat b_{k\mid K}}
\sum_{t=1}^{\widehat b_{k\mid K}}X_{it}.
\]

Define the within-cluster fitting criterion
\[
V(K)
=
\frac{1}{|\widehat{\mathcal C}_1|}
\sum_{k=1}^{K}
\sum_{i\in\widehat{\mathcal C}(k\mid K)}
\frac{1}{T}
\sum_{t=1}^{T}
\left\|
X_{it}-\widehat\nu_{it,k\mid K}
\right\|_{L^2(C)}^2,
\]
and the information criterion
\begin{equation}
IC(K)=\log V(K)+K\rho_{NT},
\label{eq:3.17}
\end{equation}
where $\rho_{NT}$ is a vanishing penalty sequence. The selected number of clusters is
\begin{equation}
\widehat K=\arg\min_{1\le K\le \overline K}IC(K),
\label{eq:3.18}
\end{equation}
where $\overline K$ is a pre-specified upper bound.

\begin{assumption}[Latent grouping and information-criterion penalty]
\label{ass:cluster}
\normalfont
Suppose that the latent structure \eqref{eq:3.14} holds and:
\begin{enumerate}[label=(\arabic*),leftmargin=2.2em,itemsep=0.35em]
\item There exist constants $0<c_1<\cdots<c_{K_0}<1$ such that
\[
b_k=\lfloor c_k T\rfloor,
\qquad
k=1,\ldots,K_0.
\]

\item The cluster sizes satisfy
\[
|\mathcal C(b_k)|=d_k|\mathcal C_1|,
\qquad
d_k>0,
\qquad
\sum_{k=1}^{K_0}d_k=1.
\]

\item The signal magnitudes are uniformly bounded and separated:
\[
\sup_{1\le i\le N}\|\mu_i\|_{L^2(C)}<\infty,
\qquad
\sup_{i\in\mathcal C_1}\|\delta_i\|_{L^2(C)}<\infty,
\qquad
\inf_{i\in\mathcal C_1}\|\delta_i\|_{L^2(C)}>0.
\]

\item The penalty sequence satisfies
\[
\rho_{NT}\to0,
\qquad
\frac{T\rho_{NT}}{[\log(N\vee T)]^{1+\zeta}}\to\infty
\]
for every arbitrarily small $\zeta>0$.

\item The candidate range contains the truth and is fixed, \(K_0\le\overline K<\infty\). Moreover, the fitting criterion separates underfitting from the oracle partition on the log-loss scale and makes overfitting gains negligible relative to the penalty:
\[
\inf_{1\le K<K_0}\{\log V(K)-\log V(K_0)\}\ge c+o_P(1)
\]
for some constant \(c>0\), and
\[
\sup_{K_0<K\le\overline K}
|\log V(K)-\log V(K_0)|
=o_P(\rho_{NT}).
\]
\end{enumerate}
\end{assumption}

Assumption \ref{ass:cluster} formalizes the latent grouping structure. Conditions (1)--(2) impose separation of the true break locations and nondegenerate cluster sizes. Condition (3) rules out vanishing or exploding signals within the changed set. Conditions (4)--(5) are the information-criterion requirements: the penalty vanishes, but remains large enough to remove overfitted groups, while underfitted partitions retain a non-negligible approximation loss.

\begin{theorem}[Recovery of latent break groups]
\label{thm:3.6}
Suppose that Assumptions \ref{ass:regularity}, \ref{ass:screening}, \ref{ass:localmax}, and \ref{ass:cluster} hold. Then, as $N,T\to\infty$ jointly,
\begin{equation}
P(\widehat K=K_0)\to1.
\label{eq:3.19}
\end{equation}
Moreover,
\begin{equation}
P\!\left(
\widehat{\mathcal C}(b_k)=\mathcal C(b_k),\ k=1,\ldots,K_0
\ \big|\
\widehat K=K_0
\right)\to1,
\label{eq:3.20}
\end{equation}
where
\[
\widehat{\mathcal C}(b_k):=\widehat{\mathcal C}(k\mid\widehat K).
\]
\end{theorem}

Theorem \ref{thm:3.6} shows that the latent grouping structure can be recovered once the changed subjects and their preliminary break locations are estimated accurately. The global test, subject screening rule, and gap-based clustering step therefore form a coherent inferential pipeline.

After latent groups have been recovered, the break-point estimate within each group can be refined by pooling the subject-wise CUSUM energies. For \(k=1,\ldots,\widehat K\), write
\[
\widehat{\mathcal C}_k=\widehat{\mathcal C}(k\mid\widehat K).
\]
Define
\begin{equation}
\widetilde b_k
=
\arg\max_{1\le t\le T}
\sum_{i\in\widehat{\mathcal C}_k}
\int_C C_{i,T}^2\!\left(\frac{t}{T},u\right)\,du,
\qquad
k=1,\ldots,\widehat K.
\label{eq:3.21}
\end{equation}

\begin{theorem}[Post-clustering pooled localization]
\label{thm:3.7}
Suppose that the latent structure \eqref{eq:3.14}, Assumption \ref{ass:regularity}, Assumption \ref{ass:cluster}, and the conditions of Theorem \ref{thm:3.6} hold. In addition, assume that
\[
|\mathcal C_1|=O(T^2),\qquad T=O(|\mathcal C_1|^{3/2}),
\]
and
\begin{equation}
\min_{1\le k\le K_0}
\frac{1}{|\mathcal C(b_k)|^{1/2}}
\sum_{i\in\mathcal C(b_k)}
\|\delta_i\|_{L^2(C)}^2
\to\infty.
\label{eq:3.22}
\end{equation}
Then, as $T$ and $|\mathcal C_1|$ tend to infinity jointly,
\begin{equation}
\max_{1\le k\le K_0}
|\widetilde b_k-b_k|
=o_P(T).
\label{eq:3.23}
\end{equation}
\end{theorem}

Theorem \ref{thm:3.7} shows that, after group recovery, pooling information within each estimated cluster yields consistent group-level break localization. Exact integer recovery would require stronger signal-to-noise conditions; the displayed rate is the baseline conclusion under the stated assumptions.

\begin{remark}[Growth conditions in Theorem \ref{thm:3.7}]
The growth restrictions involving \(T\) and \(|\mathcal C_1|\) control the pooled stochastic fluctuation after cluster recovery. They ensure that the information gained by averaging within a recovered cluster is large enough for the deterministic pooled CUSUM signal to dominate the residual noise.
\end{remark}

\section{Simulation Studies}
\label{sec:simulation}

\subsection{Data-Generating Process}
\label{subsec:data-generation}

We consider high-dimensional functional observations on the domain $\mathcal U=[0,1]$. For each component $i=1,\dots,N$ and time point $t=1,\dots,T$, let $X_{i,t}(u)$, $u\in[0,1]$, denote the observed curve. To comprehensively assess the performance of the proposed procedure, we consider several data-generating mechanisms that vary along three dimensions: the functional generation scheme, the jump-direction pattern, and the signal sparsity level.

\paragraph{(I) Functional generation mechanism.}


Motivated by practical functional time series such as daily load and renewable energy generation curves, we consider two representative models: (i) a basis-expansion model with dependent latent coefficients, and (ii) an SDE-based model with smooth stochastic trajectories. These two designs capture complementary sources of functional variation and provide distinct benchmarks for evaluating change-point detection performance.

\paragraph{Model I: basis-expansion model.}
The functional observations are generated through a Fourier basis expansion,
\[
X_{i,t}(u)=\sum_{j=1}^{J}\xi_{i,t,j}\phi_j(u), \qquad u\in[0,1],
\]
where $\{\phi_j(\cdot)\}_{j=1}^{J}$ is a Fourier basis system. The latent coefficients are generated as
\[
\xi_{i,t,j}=\beta_{i,t,j}+\eta_{i,t,j}, \qquad i=1,\dots,N,\; t=1,\dots,T,\; j=1,\dots,J.
\]
For each $j$, let
\[
\bm{\beta}_{t,j}=(\beta_{1,t,j},\dots,\beta_{N,t,j})^\top.
\]
We assume that $\bm{\beta}_{t,j}$ follows a VAR(1) model,
\[
\bm{\beta}_{t,j}=A\bm{\beta}_{t-1,j}+\bm{\varepsilon}_{t,j}, \qquad t\ge2,
\]
where $\bm{\varepsilon}_{t,j}\sim N(\bm 0,I_N)$ independently across $t$ and $j$, and $\bm{\beta}_{1,j}\sim N(\bm 0,I_N)$. The transition matrix $A=(a_{kl})_{1\le k,l\le N}$ is banded:
\[
a_{kl}\sim \mathrm{Unif}(-0.3,0.3)\;\; \text{if } |k-l|\le 3, \qquad
a_{kl}=0\;\; \text{if } |k-l|>3.
\]
This induces both local cross-sectional dependence and temporal dependence. The idiosyncratic term is generated independently as
\[
\eta_{i,t,j}\sim N(0,1/j),
\]
so that higher-order basis coefficients have smaller variances.

\paragraph{Model II: SDE-based model.}
As an alternative functional generation scheme, we generate the baseline curves from a stochastic differential equation. Specifically, for each $(i,t)$, let $Z_{i,t}(u)$ solve
\[
dZ_{i,t}(u)=-\lambda Z_{i,t}(u)\,du+\sigma\, dW_{i,t}(u), \qquad u\in[0,1],
\]
where $W_{i,t}(\cdot)$ is a standard Wiener process, and $\lambda,\sigma>0$ are fixed constants. The observed curve is then generated as
\[
X_{i,t}(u)=\mu_{i,t}(u)+Z_{i,t}(u),
\]
where $\mu_{i,t}(u)$ is the mean function that may change over time. This model produces smoother trajectories and serves as a complementary benchmark to the basis-expansion setting.

\paragraph{(II) Null and alternative hypotheses.}
Under the null hypothesis $H_0$, no change-point is present in any component, and the observations are generated solely from the baseline process. That is,
\[
X_{i,t}(u)=\varepsilon_{i,t}(u), \qquad i=1,\dots,N,\; t=1,\dots,T,
\]
for Model I, where
\[
\varepsilon_{i,t}(u)=\sum_{j=1}^{J}(\beta_{i,t,j}+\eta_{i,t,j})\phi_j(u),
\]
and
\[
X_{i,t}(u)=Z_{i,t}(u), \qquad i=1,\dots,N,\; t=1,\dots,T,
\]
for Model II.

Under the alternative hypothesis $H_1$, only a subset of components undergo mean shifts. Let
\[
\mathcal I_{\mathrm{cp}}\subset\{1,\dots,N\}, \qquad |\mathcal I_{\mathrm{cp}}|=\lfloor \mathrm{SDR}\cdot N\rfloor,
\]
where $\mathrm{SDR}\in(0,1]$ controls the proportion of changed components. For each $i\in\mathcal I_{\mathrm{cp}}$, the change-point location $\tau_i$ is drawn independently from
\[
\tau_i \sim \mathrm{Unif}\bigl\{\lfloor 0.25T\rfloor,\dots,\lfloor 0.75T\rfloor\bigr\},
\]
whereas $\tau_i=\infty$ for $i\notin\mathcal I_{\mathrm{cp}}$.

\paragraph{(III) Signal shape and sparsity.}
The mean-shift function is defined as
\[
\delta^\star(u)=\frac{1}{\sqrt m}\sum_{j=1}^{m}\phi_j(u),
\]
where $m$ controls the number of affected basis directions. Hence, smaller $m$ corresponds to a basis-sparse signal, while larger $m$ produces a basis-dense signal.



Each data model includes a mean shift after the component-specific change-point location, with two jump-direction patterns: (i) one-sided jumps, representing aligned changes across components, and (ii) mixed-sign jumps, designed to introduce cross-sectional signal cancellation.

\paragraph{Type A: one-sided jumps.}
All changed components share the same jump direction:
\[
\delta_i(u)=
\begin{cases}
\sqrt{\mathrm{SNR}}\,\delta^\star(u), & i\in\mathcal I_{\mathrm{cp}},\\[0.3em]
0, & i\notin\mathcal I_{\mathrm{cp}}.
\end{cases}
\]

\paragraph{Type B: mixed-sign jumps.}
To introduce cancellation effects, we partition the changed set as
\[
\mathcal I_{\mathrm{cp}}=\mathcal I_{\mathrm{cp}}^{+}\cup \mathcal I_{\mathrm{cp}}^{-}, \qquad \mathcal I_{\mathrm{cp}}^{+}\cap \mathcal I_{\mathrm{cp}}^{-}=\varnothing,
\]
with $|\mathcal I_{\mathrm{cp}}^{+}|\approx |\mathcal I_{\mathrm{cp}}^{-}|$, and define
\[
\delta_i(u)=
\begin{cases}
\sqrt{\mathrm{SNR}}\,\delta^\star(u), & i\in \mathcal I_{\mathrm{cp}}^{+},\\[0.3em]
-\sqrt{\mathrm{SNR}}\,\delta^\star(u), & i\in \mathcal I_{\mathrm{cp}}^{-},\\[0.3em]
0, & i\notin \mathcal I_{\mathrm{cp}}.
\end{cases}
\]

The observed process under $H_1$ is generated by
\[
X_{i,t}(u)=
\begin{cases}
\varepsilon_{i,t}(u), & t\le \tau_i,\\[0.3em]
\varepsilon_{i,t}(u)+\delta_i(u), & t>\tau_i,
\end{cases}
\qquad i=1,\dots,N,\; t=1,\dots,T,
\]
for Model I, and by
\[
X_{i,t}(u)=
\begin{cases}
Z_{i,t}(u), & t\le \tau_i,\\[0.3em]
Z_{i,t}(u)+\delta_i(u), & t>\tau_i,
\end{cases}
\qquad i=1,\dots,N,\; t=1,\dots,T,
\]
for Model II.

\paragraph{(IV) Sparse and dense regimes.}
We examine both component-wise sparse and dense contamination settings through the parameter $\mathrm{SDR}$. Small values of $\mathrm{SDR}$ correspond to sparse alternatives, where only a few components change, whereas large values of $\mathrm{SDR}$ represent dense alternatives. Combined with the choice of $m$, this yields four representative signal regimes:


\begin{table}[htbp]
\centering
\caption{Simulation scenarios for assessing sparsity and sign heterogeneity.}
\label{tab:simulation_scenarios}
\begin{threeparttable}
\setlength{\tabcolsep}{5pt}
\renewcommand{\arraystretch}{1.15}
\begin{tabular}{ccccccccc}
\toprule
\multirow{2}{*}{\textbf{Scenario}} 
& \multicolumn{2}{c}{\textbf{Mechanism}}
& \multicolumn{2}{c}{\textbf{Jump type}}
& \multicolumn{2}{c}{\textbf{Comp. sparsity}}
& \multicolumn{2}{c}{\textbf{Basis sparsity}} \\
\cmidrule(lr){2-3} \cmidrule(lr){4-5} \cmidrule(lr){6-7} \cmidrule(lr){8-9}
& Basis & SDE 
& One-sided & Mixed-sign 
& Dense & Sparse 
& Dense & Sparse \\
\midrule
S1 & $\checkmark$ &      & $\checkmark$ &      & $\checkmark$ &      & $\checkmark$ &  \\
S2 & $\checkmark$ &      &              & $\checkmark$ &              & $\checkmark$ &              & $\checkmark$ \\
S3 &              & $\checkmark$ &              & $\checkmark$ &              & $\checkmark$ & $\checkmark$ &  \\
S4 &              & $\checkmark$ & $\checkmark$ &              & $\checkmark$ &              &              & $\checkmark$ \\
\bottomrule
\end{tabular}
\begin{tablenotes}[flushleft]
\footnotesize
\item[$\dagger$] A check mark indicates that the corresponding feature is present in the scenario, while a blank entry indicates that the feature is absent. ``Comp. sparsity'' refers to the component-wise sparsity level controlled by SDR, and ``Basis sparsity'' refers to the number of affected basis directions controlled by \(m\).
\end{tablenotes}
\end{threeparttable}
\end{table}

In implementation, the functional trajectories are evaluated on an equally spaced grid of size \texttt{grid\_size} over $[0,1]$. Unless otherwise specified, the main tuning parameters are the number of basis functions $J$, the signal proportion $\mathrm{SDR}$, the number of affected basis directions $m$, and the signal strength $\mathrm{SNR}$.


Unless otherwise stated, we set $N=100$, $T=200$, $J=15$, and $\texttt{grid\_size}=50$. 

To further assess performance in genuine multiple change-point settings, we also consider a multi-jump extension in which each changed component may experience more than one structural break. Specifically, for $i\in\mathcal I_{\mathrm{cp}}$, let
\[
X_{i,t}(u)=\varepsilon_{i,t}(u)+\sum_{\ell=1}^{K_i}\delta_{i,\ell}(u)\mathbf 1(t>\tau_{i,\ell}),
\]
where $K_i\in\{1,2\}$ is randomly generated, and the break locations $\tau_{i,1}<\cdots<\tau_{i,K_i}$ are sampled with a minimum spacing constraint. The jump directions $\delta_{i,\ell}(u)$ may be either common-sign or mixed-sign.

\subsection{Size Performance}
\label{subsec:size-performance}

We first examine the empirical size of the competing tests under the null hypothesis of no structural break. Since the structural-break components are removed under $H_0$, the four alternative scenarios reduce to two distinct null-generation mechanisms: the basis-expansion null corresponding to S1/S2 and the SDE-based null corresponding to S3/S4. We compare the conventional CUSUM statistic, the PE--CUSUM statistic, and the proposed Energy--PE statistic. The rejection frequencies are computed over Monte Carlo replications at the nominal significance levels $\alpha=0.01, 0.05$, and $0.10$.

\begin{table}[htbp]
\centering
\caption{Empirical Size under Two Null Data-Generating Mechanisms}
\label{tab:size_all}
\small
\setlength{\tabcolsep}{3.5pt}
\renewcommand{\arraystretch}{1.12}
\resizebox{\textwidth}{!}{%
\begin{tabular}{llcccccccc}
\toprule
\multirow{2}{*}{\textbf{Level}} & \multirow{2}{*}{\textbf{Method}}
& \multicolumn{2}{c}{\textbf{\(N=100\)}}
& \multicolumn{2}{c}{\textbf{\(N=200\)}}
& \multicolumn{2}{c}{\textbf{\(N=500\)}}
& \multicolumn{2}{c}{\textbf{\(N=800\)}} \\
\cmidrule(lr){3-4} \cmidrule(lr){5-6} \cmidrule(lr){7-8} \cmidrule(lr){9-10}
& & \textbf{Basis} & \textbf{SDE} & \textbf{Basis} & \textbf{SDE} & \textbf{Basis} & \textbf{SDE} & \textbf{Basis} & \textbf{SDE} \\
\midrule
\multirow{3}{*}{$\alpha=0.01$}
& CUSUM     & 0.008 & 0.006 & 0.013 & 0.007 & 0.009 & 0.011 & 0.004 & 0.030 \\
& PE--CUSUM & 0.006 & 0.006 & 0.012 & 0.007 & 0.009 & 0.011 & 0.004 & 0.030 \\
& Energy--PE & 0.004 & 0.008 & 0.016 & 0.013 & 0.008 & 0.016 & 0.005 & 0.014 \\
\midrule
\multirow{3}{*}{$\alpha=0.05$}
& CUSUM     & 0.032 & 0.050 & 0.059 & 0.054 & 0.038 & 0.041 & 0.038 & 0.061 \\
& PE--CUSUM & 0.032 & 0.050 & 0.061 & 0.054 & 0.057 & 0.041 & 0.063 & 0.061 \\
& Energy--PE & 0.029 & 0.023 & 0.061 & 0.042 & 0.048 & 0.055 & 0.070 & 0.054 \\
\midrule
\multirow{3}{*}{$\alpha=0.10$}
& CUSUM     & 0.076 & 0.116 & 0.089 & 0.119 & 0.106 & 0.097 & 0.083 & 0.109 \\
& PE--CUSUM & 0.112 & 0.116 & 0.109 & 0.119 & 0.084 & 0.097 & 0.134 & 0.109 \\
& Energy--PE & 0.115 & 0.077 & 0.116 & 0.103 & 0.097 & 0.108 & 0.146 & 0.113 \\
\bottomrule
\end{tabular}%
}
\end{table}

Table~\ref{tab:size_all} reports the empirical sizes of the CUSUM, PE--CUSUM, and Energy--PE tests under the two null data-generating mechanisms, namely the basis-expansion null corresponding to S1/S2 and the SDE-based null corresponding to S3/S4. Overall, the three methods exhibit reasonably controlled empirical sizes across the considered sample sizes and nominal significance levels.

For $N=100$, the empirical sizes under the basis-expansion null are slightly conservative at $\alpha=0.01$ and $\alpha=0.05$, with rejection frequencies below the nominal levels for all three methods. At $\alpha=0.10$, the CUSUM test remains conservative, whereas PE--CUSUM and Energy--PE are closer to the target level. Under the SDE-based null, CUSUM and PE--CUSUM match the nominal level well at $\alpha=0.05$, while Energy--PE is somewhat conservative. At $\alpha=0.10$, all methods remain reasonably close to the nominal level, with Energy--PE showing a slightly smaller rejection rate.

When $N$ increases to 200, the empirical sizes become closer to the nominal levels in most cases. Under the basis-expansion null, all three methods are well controlled at $\alpha=0.01$ and $\alpha=0.05$, and PE--CUSUM and Energy--PE are close to the nominal level at $\alpha=0.10$. Under the SDE-based null, CUSUM and PE--CUSUM show nearly identical sizes, while Energy--PE is slightly more conservative, especially at $\alpha=0.05$. These results suggest that the proposed Energy--PE statistic does not exhibit evident size inflation under the null and maintains comparable size performance to the benchmark CUSUM and PE--CUSUM procedures.

\subsection{Power Performance}
\label{subsec:power-performance}

We next investigate empirical power under four representative alternatives. Unlike the size experiment, the structural-break components are retained here. Scenario S1 corresponds to dense one-sided breaks under the basis-expansion mechanism, S2 to sparse mixed-sign breaks under the same mechanism, S3 to sparse mixed-sign breaks under the SDE-based mechanism, and S4 to dense one-sided breaks under the SDE-based mechanism. Together, these settings assess sensitivity to sparsity, sign heterogeneity, and the baseline dependence structure of the functional observations.

We compare the conventional CUSUM test, PE--CUSUM, and the proposed Energy--PE test. The rejection frequencies are computed over Monte Carlo replications at the nominal level $\alpha=0.05$. Since Energy--PE aggregates subject-specific CUSUM energies before cross-sectional averaging while retaining a power-enhancement component, it is designed to remain sensitive under both dense sign-heterogeneous alternatives and sparse strong-signal alternatives.

\begin{table}[htbp]
\centering
\caption{Empirical Power under Dense, Sparse, and Mixed-Sign Alternatives. Results are based on $n_{\mathrm{rep}}=300$ Monte Carlo replications with $N=100$, $T=200$, $\mathrm{SNR}=1$, $\mathrm{SDR}_{\mathrm{dense}}=0.8$, and $\mathrm{SDR}_{\mathrm{sparse}}=0.1$.}
\label{tab:power_four_scenarios}
\setlength{\tabcolsep}{8pt}
\renewcommand{\arraystretch}{1.15}
\begin{tabular}{cccc}
\toprule
\textbf{Scenario}  & \textbf{CUSUM} & \textbf{PE--CUSUM} & \textbf{Energy--PE} \\
\midrule
S1    & 1.000 & 1.000 & 1.000 \\
S2 & 0.043 & 0.587 & 0.620 \\
S3   & 0.150 & 0.150 & 1.000 \\
S4    & 1.000 & 1.000 & 1.000 \\
\bottomrule
\end{tabular}
\end{table}

Table~\ref{tab:power_four_scenarios} reports the empirical powers of CUSUM, PE--CUSUM, and Energy--PE under the four alternative scenarios. In the dense one-sided scenarios S1 and S4, all methods achieve power close to one, indicating that these alternatives are relatively easy to detect and that the conventional CUSUM statistic is already effective. Clear differences emerge under the sparse mixed-sign scenarios. In S2, CUSUM has very low power because sign cancellation weakens the cross-sectional average, whereas PE--CUSUM and Energy--PE substantially improve rejection probability, with Energy--PE giving the highest power. The improvement is more pronounced in S3, where Energy--PE attains power equal to one while CUSUM and PE--CUSUM remain much less sensitive. These results indicate that Energy--PE preserves high power under dense aligned alternatives and provides clear gains under sparse or sign-heterogeneous alternatives.

\begin{figure}[htbp]
    \centering
    \includegraphics[width=0.95\textwidth]{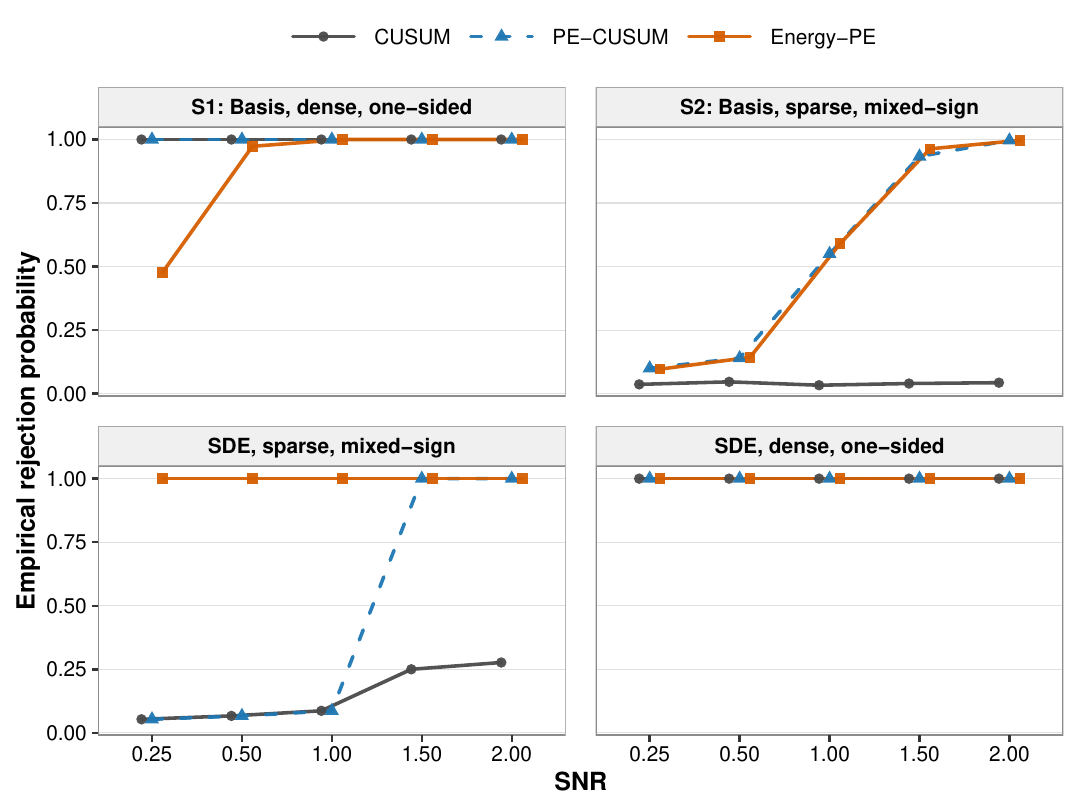}
    \caption{Energy--PE maintains high rejection probability under weak mixed-sign signals. Empirical rejection probabilities are reported for CUSUM, PE--CUSUM, and Energy--PE as the signal-to-noise ratio varies across the four alternative scenarios. The nominal significance level is $\alpha=0.05$, with $N=100$ and $T=200$.}
    \label{fig:power_comparison_snr}
\end{figure}

Figure~\ref{fig:power_comparison_snr} reports the empirical rejection probabilities as the signal-to-noise ratio varies. In the dense one-sided scenarios S1 and S4, all methods quickly attain high power, indicating that these alternatives are relatively easy to detect. The advantage of the proposed Energy--PE becomes more evident in the sparse mixed-sign scenarios S2 and S3, where the conventional CUSUM test suffers from power loss due to sign cancellation in the cross-sectional average. In S2, both PE--CUSUM and Energy--PE improve substantially over CUSUM, with Energy--PE giving slightly higher power in most cases. In S3, Energy--PE shows a particularly strong advantage and maintains high power even under weak signals. These results suggest that Energy--PE preserves competitive performance in dense aligned alternatives while providing clear gains under more challenging sign-heterogeneous settings.

\subsection{Post-Test Identification and Break-Location Estimation}
\label{subsec:post-test-identification}

After rejecting the global null hypothesis, we further identify the subjects that contain structural breaks and estimate their break locations. For the proposed Energy--PE procedure, this post-test step uses the quantities introduced in Section~\ref{subsec:sparse-enhancement}: the subject-specific energy score \(Y_{iT}\) in \eqref{eq:subject_energy}, the standardized score \(S_{iT}\) in \eqref{eq:standardized_energy}, and the generalized enhancement component \(\mathcal Z_{NT}^{g}\) in \eqref{eq:enhancement_general}. We write \(g_{iT}=g(S_{iT})\) for the transformed subject-level score.

For subject identification, we threshold the subject-level energy scores. In the hard-threshold case $g(z)=\mathbf 1\{z>0\}$, the estimated set of subjects with breaks is
\[
\widehat{\mathcal C}_{\bullet}^{\,E}
=
\left\{
1\le i\le N:
Y_{iT}>\xi_{NT}
\right\}.
\]
More generally, for a monotone transformation $g$, we use
\[
\widehat{\mathcal C}_{\bullet}^{\,g}
=
\left\{
1\le i\le N:
g_{iT}>g(0)
\right\}.
\]
For the monotone choices considered in this study, including the indicator, positive-part, softplus, and sigmoid transformations, this rule is equivalent to $Y_{iT}>\xi_{NT}$. Hence, the subject-screening step is driven by the same underlying CUSUM energy score, while different choices of $g$ mainly affect the global enhancement statistic. This also implies that the post-test screening mechanism is closely related to the subject-level screening used in PE--CUSUM; the main methodological difference between PE--CUSUM and Energy--PE lies in the construction of the global test statistic, where Energy--PE replaces the mean-aggregated CUSUM component with a quadratic subject-wise energy aggregation.

For each detected subject $i\in\widehat{\mathcal C}_{\bullet}^{\,g}$, the individual break location is estimated by
\[
\widehat{\tau}_{i}^{\,E}
=
\arg\max_{1\le t<T}
\int_{\mathcal C}
C_{i,T}^2(t/T,u)\,du .
\]
Thus, the same subject-specific energy path is used both to determine whether a subject contains a break and to locate the break point.

We evaluate the subject-identification accuracy by treating subjects with breaks as class one and subjects without breaks as class zero. Let TP, FP, and FN denote the numbers of true positives, false positives, and false negatives, respectively. We report the true positive rate (TPR), precision, and the $F_1$ score, defined as
\[
\mathrm{TPR}
=
\frac{\mathrm{TP}}{\mathrm{TP}+\mathrm{FN}},
\qquad
\mathrm{Precision}
=
\frac{\mathrm{TP}}{\mathrm{TP}+\mathrm{FP}},
\]
and
\[
F_1
=
\frac{2\,\mathrm{Precision}\cdot \mathrm{TPR}}
{\mathrm{Precision}+\mathrm{TPR}}
=
\frac{\mathrm{TP}}
{\mathrm{TP}+(\mathrm{FP}+\mathrm{FN})/2}.
\]
Here, TPR measures the proportion of truly changed subjects that are successfully detected, whereas precision measures the proportion of detected subjects that are truly changed. The accuracy of break-location estimation is measured by the mean squared distance between the estimated and true break locations,
\[
\mathrm{MSD}_{\tau}
=
\frac{1}{|\mathcal C_{\bullet}|}
\sum_{i\in\mathcal C_{\bullet}}
(\widehat \tau_i/T-\tau_i/T)^2\times10^3,
\]
where $\mathcal C_{\bullet}$ denotes the set of subjects with true structural breaks.

In the main scenarios S1--S4, the break locations are generated independently across changed subjects. Therefore, these scenarios are used to assess global power, subject identification, and individual break-location estimation. To evaluate the latent common-break structure, we additionally consider a separate setting in which changed subjects share a small number of common break points. In this setting, we cluster the estimated break locations by a gap-based procedure, select the number of latent groups using an information criterion, and report the accuracy of group-number estimation, Purity, normalized mutual information (NMI), and the post-clustering break-location error.

Let
\[
Y_{iT}
=
\max_{1\le t<T}
\int_0^1 C_{iT}^2(t/T,u)\,du
\]
denote the subject-level CUSUM energy statistic. For each Monte Carlo
sample generated under the null hypothesis, denoted by
\(b=1,\ldots,B\), we compute
\[
Y_{iT}^{0,(b)}, \qquad i=1,\ldots,N .
\]

\paragraph{Subjectwise calibration.}
The subjectwise calibrated threshold is defined as
\[
\widehat \xi_{NT}^{\mathrm{sub}}(\alpha)
=
\widehat Q_{1-\alpha}
\left(
\left\{
Y_{iT}^{0,(b)}
:
i=1,\ldots,N,\ b=1,\ldots,B
\right\}
\right),
\]
where \(\widehat Q_{1-\alpha}(\cdot)\) denotes the empirical
\((1-\alpha)\)-quantile. This threshold approximately controls the
false positive probability for each individual subject:
\[
\mathbb P_{H_0}
\left(
Y_{iT} > \widehat \xi_{NT}^{\mathrm{sub}}(\alpha)
\right)
\approx \alpha .
\]

\paragraph{Familywise max calibration.}
The familywise calibrated threshold is defined as
\[
\widehat \xi_{NT}^{\mathrm{fam}}(\alpha)
=
\widehat Q_{1-\alpha}
\left(
\left\{
\max_{1\le i\le N} Y_{iT}^{0,(b)}
:
b=1,\ldots,B
\right\}
\right).
\]
This threshold approximately controls the probability of selecting at
least one false positive subject under the null:
\[
\mathbb P_{H_0}
\left(
\max_{1\le i\le N}Y_{iT}
>
\widehat \xi_{NT}^{\mathrm{fam}}(\alpha)
\right)
\approx \alpha .
\]

Figure~\ref{fig:calibration_performance} summarizes the effect of the
two H0-calibrated thresholding rules on post-test subject identification.
The subjectwise calibration yields better overall $F_1$ scores in the
basis-expansion settings S1 and S2, while the familywise max calibration
performs particularly well in the SDE-based settings S3 and S4. The
TPR--Precision plot shows the underlying tradeoff: subjectwise calibration
increases recall, whereas familywise max calibration imposes stronger
false-positive control and therefore yields higher precision.

\begin{figure}[htbp]
\centering

\begin{subfigure}{0.82\textwidth}
\centering
\includegraphics[width=\textwidth]{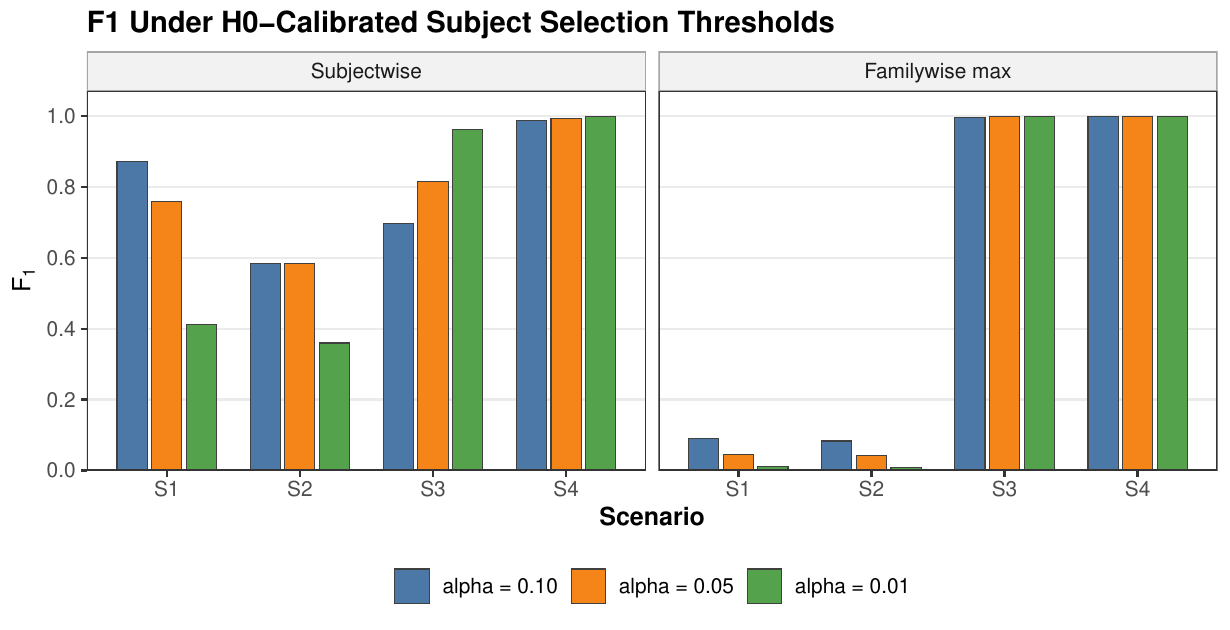}
\caption{$F_1$ score}
\label{fig:calibration_f1}
\end{subfigure}

\vspace{0.8em}

\begin{subfigure}{0.82\textwidth}
\centering
\includegraphics[width=\textwidth]{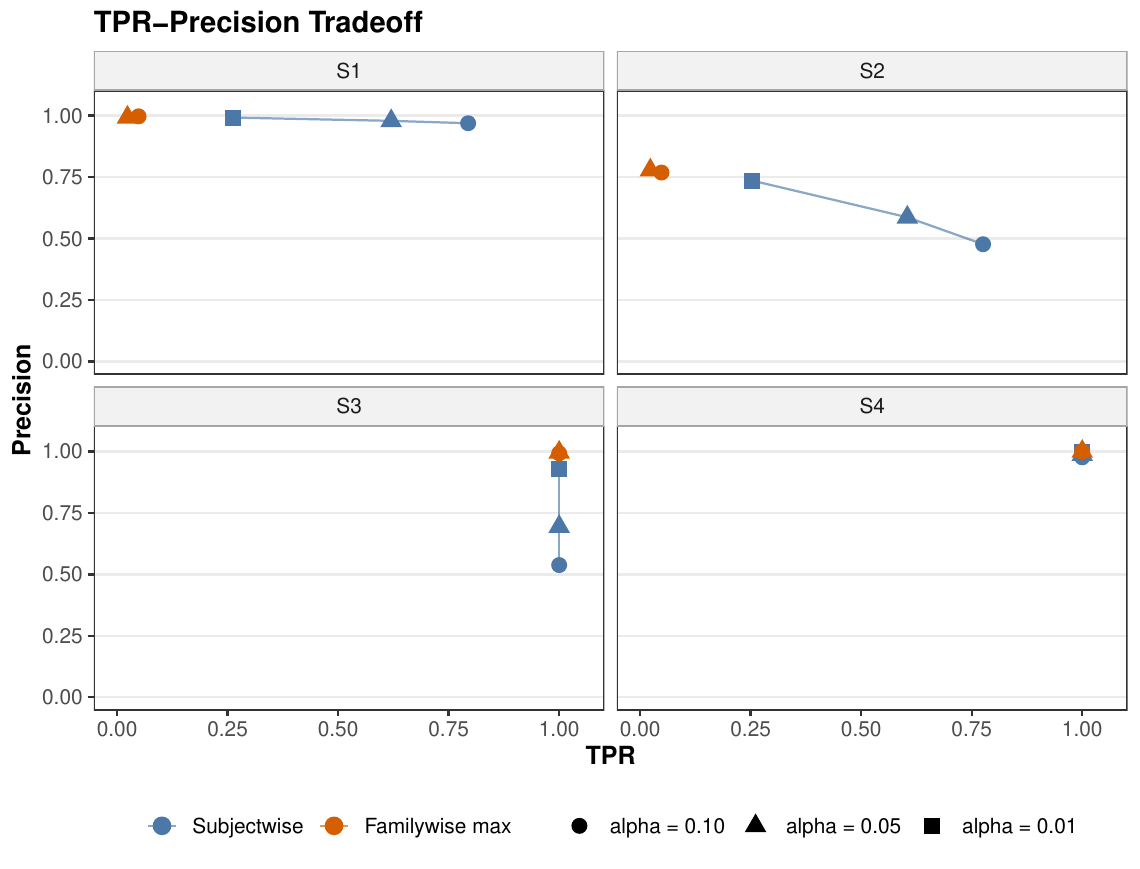}
\caption{TPR--precision tradeoff}
\label{fig:calibration_tpr_precision}
\end{subfigure}

\caption{Subjectwise calibration improves recall, whereas familywise max calibration improves precision. Post-test subject-identification performance is compared using $F_1$ scores and the TPR--precision tradeoff.}
\label{fig:calibration_performance}
\end{figure}

\begin{table}[htbp]
\centering
\caption{Post-Test Identification Accuracy across Signal-to-Noise Ratios. Results use $g(z)=\mathbf{1}\{z>0\}$, $N=100$, $T=200$, and $n_{\mathrm{rep}}=300$.}
\label{tab:post_identification_snr_h0_subjectwise}
\setlength{\tabcolsep}{6pt}
\renewcommand{\arraystretch}{1.15}
\begin{tabular}{lccccc}
\toprule
\multirow{2}{*}{\textbf{Scenario}} & \multirow{2}{*}{\textbf{Metric}}
& \multicolumn{4}{c}{\textbf{SNR}} \\
\cmidrule(lr){3-6}
& & 0.25 & 0.50 & 1.00 & 2.00 \\
\midrule
\multirow{2}{*}{S1} & $F_1$ & 0.212 & 0.377 & 0.762 & 0.980 \\
                    & MSD$_\tau$ & 19.166 & 11.819 & 5.143 & 2.046 \\
\midrule
\multirow{2}{*}{S2} & $F_1$ & 0.143 & 0.271 & 0.598 & 0.803 \\
                    & MSD$_\tau$ & 15.975 & 10.785 & 6.153 & 2.410 \\
\midrule
\multirow{2}{*}{S3} & $F_1$ & 0.816 & 0.821 & 0.826 & 0.816 \\
                    & MSD$_\tau$ & 0.030 & 0.007 & 0.001 & 0.000 \\
\midrule
\multirow{2}{*}{S4} & $F_1$ & 0.994 & 0.993 & 0.994 & 0.993 \\
                    & MSD$_\tau$ & 0.101 & 0.023 & 0.005 & 0.001 \\
\bottomrule
\end{tabular}
\end{table}

The pre-clustering mean squared distance (Pre-MSD) measures the localization error of the individual break estimates before applying the latent-group structure:
\[
\mathrm{Pre\text{-}MSD}
=
\frac{1}{|\widehat{\mathcal C}_{\bullet}\cap\mathcal C_{\bullet}|}
\sum_{i\in\widehat{\mathcal C}_{\bullet}\cap\mathcal C_{\bullet}}
(\widehat{\tau}_i-\tau_i)^2.
\]
The post-clustering mean squared distance (Post-MSD) measures the localization error after grouping subjects with common break points and re-estimating the group-level breaks by pooled CUSUM energy:
\[
\mathrm{Post\text{-}MSD}
=
\frac{1}{K_0}
\sum_{k=1}^{K_0}
(\widehat b_{(k)}-b_{(k)})^2,
\]
where \(b_{(k)}\) and \(\widehat b_{(k)}\) are the ordered true and estimated common break points, respectively. Smaller values of Pre-MSD and Post-MSD indicate more accurate break-location estimation.

\paragraph{Latent common-break setting.}
To further examine whether the proposed procedure can recover latent common-break structures, we consider an additional simulation setting. In contrast to the main scenarios S1--S4, where the break locations are generated independently across changed subjects, this setting assumes that changed subjects share a small number of common break points. Specifically, we set \(K_0=3\) and choose the common break locations as
\[
b_1=\lfloor 0.25T\rfloor,\qquad
b_2=\lfloor 0.50T\rfloor,\qquad
b_3=\lfloor 0.75T\rfloor .
\]
A proportion of subjects are generated without structural breaks, while the changed subjects are evenly divided into three latent groups, each associated with one of the common break points. The proposed Energy--PE procedure is first used to identify subjects with breaks and estimate their individual break locations. The estimated break locations are then clustered by the gap-based procedure, and the number of latent groups is selected by the information criterion described above.

We evaluate the latent-structure recovery by reporting the probability of correctly estimating the number of groups, \(\Pr(\widehat K=K_0)\), together with Purity and normalized mutual information (NMI) for group membership recovery. To assess localization accuracy, we report both the pre-clustering mean squared distance (Pre-MSD), based on the individual estimates \(\widehat\tau_i\), and the post-clustering mean squared distance (Post-MSD), based on the pooled group-level estimates \(\widehat b_k\).

\begin{table}[htbp]
\centering
\caption{Recovery of Latent Common-Break Groups by Energy--PE}
\label{tab:latent_structure}
\setlength{\tabcolsep}{6pt}
\renewcommand{\arraystretch}{1.15}
\begin{tabular}{lccccc}
\toprule
\textbf{SNR} 
& $\Pr(\widehat K=K_0)$ 
& \textbf{Purity} 
& \textbf{NMI} 
& \textbf{Pre-MSD} 
& \textbf{Post-MSD} \\
\midrule
0.25 & 0.143 & 0.375 & 0.055 & 31.736 & 13.101 \\
0.50 & 0.070 & 0.403 & 0.075 & 19.371 & 7.444 \\
1.00 & 0.073 & 0.511 & 0.160 & 9.447 & 1.877 \\
1.50 & 0.360 & 0.622 & 0.312 & 6.907 & 0.854 \\
2.00 & 0.747 & 0.790 & 0.550 & 5.217 & 0.355 \\
\bottomrule
\end{tabular}
\end{table}

\section{Empirical Application}
\label{sec:empirical}

This section illustrates the proposed Energy--PE procedure using Chinese A-share stock data. Each cross-sectional unit corresponds to one stock, and each time point corresponds to one intraday return curve. We use five-minute prices rather than daily closing prices because intraday trajectories provide a natural functional representation of within-day return dynamics.

The data are obtained from TuShare Pro and consist of five-minute prices for actively traded Chinese A-share stocks. To better examine the latent common-break grouping structure, we retain a relatively broad balanced panel after merging repeated downloads and removing stocks with incomplete intraday records. The final sample contains $N=22$ stocks over the period from November 14, 2022 to September 11, 2025, yielding $T=688$ trading-day functional observations. For each stock $i$ and trading day $t$, we construct the intraday cumulative return curve
\[
X_{i,t}(u_j)
=
100\{\log P_{i,t}(u_j)-\log P_{i,t}(u_1)\},
\qquad j=1,\ldots,48,
\]
where $P_{i,t}(u_j)$ denotes the five-minute price at the $j$th intraday grid point. Thus, each observation $X_{i,t}(\cdot)$ is treated as a functional object describing the within-day cumulative return pattern of stock $i$ on trading day $t$.

\begin{table}[htbp]
\centering
\caption{A-Share Intraday Functional Panel Used in the Empirical Analysis}
\label{tab:real_stock_data_description}
\setlength{\tabcolsep}{7pt}
\renewcommand{\arraystretch}{1.15}
\begin{tabular}{ll}
\toprule
\textbf{Item} & \textbf{Description} \\
\midrule
Data Source & TuShare Pro five-minute prices \\
Sample Period & 2022-11-14 to 2025-09-11 \\
Number of Stocks & $N=22$ \\
Number of Functional Observations & $T=688$ trading days \\
Functional Grid & 48 five-minute intraday positions \\
Functional Observation & $100\{\log P_{i,t}(u_j)-\log P_{i,t}(u_1)\}$ \\
\bottomrule
\end{tabular}
\end{table}

We first apply the proposed Energy--PE global test to examine whether the high-dimensional intraday functional panel contains structural changes. The null hypothesis is
\[
H_0:\quad \mu_{i,t}(u)=\mu_i(u),\qquad
i=1,\ldots,N,\quad t=1,\ldots,T,
\]
whereas the alternative allows the mean functions of a subset of stocks to change over time. Thus, the global test is used as the first-stage screening step. Only after obtaining evidence against the global null do we proceed to post-test subject selection, individual break-location estimation, and latent common-break clustering.

For the global test, we use the Energy--PE statistic with the same main threshold choice $c_\xi=0.10$ as in the post-test analysis. Critical values are computed by a centered circular block bootstrap. Specifically, we first remove the stock-specific time average from each intraday curve and then synchronously resample time blocks across all stocks, so that the bootstrap samples preserve the cross-sectional dependence and short-run temporal dependence under the null hypothesis. We use $B=1000$ bootstrap replications with block length equal to 10 trading days.

\begin{table}[htbp]
\centering
\caption{Global Energy--PE Test for Structural Changes in the A-Share Intraday Functional Panel. The critical value is obtained from a centered circular block bootstrap with $B=1000$ bootstrap replications and block length equal to 10 trading days ($\alpha=0.10$).}
\label{tab:real_stock_global_test}
\setlength{\tabcolsep}{7pt}
\renewcommand{\arraystretch}{1.15}
\begin{tabular}{ccccc}
\toprule
$c_\xi$ & Test Statistic & Critical Value & $p$-value & Decision \\
\midrule
0.10 & 210.352 & 184.292 & 0.076 & Reject $H_0$ \\
\bottomrule
\end{tabular}
\end{table}

Table~\ref{tab:real_stock_global_test} reports the first-stage global
Energy--PE test. At the nominal significance level $\alpha=0.10$, the
test statistic is 210.352, exceeding the bootstrap critical value 184.292.
The bootstrap $p$-value is 0.076. Therefore, we reject the null hypothesis
of no structural change in the high-dimensional intraday functional panel.
This provides statistical support for proceeding to the post-test analysis,
including subject selection, individual break-location estimation, and
latent common-break clustering.

We then identify which stocks contribute most strongly to the detected structural-change evidence. The subject-level screening threshold is set as
\[
\xi_{NT}
=
c_\xi \log(N\vee T)\log\log(N\vee T).
\]
Because the intraday functional observations have a different scale from the weekly curves, we report a sensitivity analysis over several smaller values of $c_\xi$. The main empirical interpretation is based on $c_\xi=0.10$, which gives a relatively conservative subject selection rule while still retaining a nontrivial latent common-break structure.

\begin{table}[htbp]
\centering
\caption{Sensitivity of A-Share Subject Selection to the Screening Threshold}
\label{tab:real_stock_cxi_sensitivity}
\setlength{\tabcolsep}{7pt}
\renewcommand{\arraystretch}{1.15}
\begin{tabular}{ccccc}
\toprule
$c_\xi$ & $\xi_{NT}$ & $\widehat{|\mathcal C|}$ & $\widehat K$ & Energy--PE Statistic \\
\midrule
0.025 & 0.307 & 18 & 1 & 472.650 \\
0.050 & 0.613 & 11 & 2 & 289.041 \\
0.075 & 0.920 & 8  & 4 & 210.352 \\
0.100 & 1.226 & 8  & 4 & 210.352 \\
0.125 & 1.533 & 5  & 3 & 131.663 \\
\bottomrule
\end{tabular}
\end{table}

Table~\ref{tab:real_stock_cxi_sensitivity} shows that the number of selected stocks decreases as $c_\xi$ increases, as expected. For very small thresholds, many stocks are selected and the estimated common-break structure is relatively coarse. For $c_\xi=0.075$ and $c_\xi=0.10$, however, the procedure selects the same number of stocks and estimates the same number of latent groups, suggesting that the grouping result is locally stable with respect to the threshold choice. We therefore use $c_\xi=0.10$ as the main specification.

\begin{table}[htbp]
\centering
\caption{Selected Stocks with the Strongest Intraday Structural Break Evidence}
\label{tab:real_stock_selected_subjects}
\setlength{\tabcolsep}{7pt}
\renewcommand{\arraystretch}{1.15}
\begin{tabular}{lccc}
\toprule
\textbf{Stock} & $\widehat\tau_i$ & \textbf{Estimated Break Date} & $Y_{iT}$ \\
\midrule
002415.SZ & 98  & 2023-04-07 & 1.588 \\
300750.SZ & 267 & 2023-12-18 & 1.775 \\
601888.SH & 269 & 2023-12-20 & 2.849 \\
600309.SH & 300 & 2024-02-02 & 1.455 \\
600048.SH & 343 & 2024-04-15 & 1.816 \\
601857.SH & 347 & 2024-04-19 & 1.451 \\
000002.SZ & 349 & 2024-04-23 & 2.170 \\
600000.SH & 493 & 2024-11-26 & 1.519 \\
\bottomrule
\end{tabular}
\end{table}

After the post-test subject selection step, we further apply the proposed latent common-break clustering procedure to the selected stocks. With $c_\xi=0.10$, the estimated number of common-break groups is $\widehat K=4$. The resulting group membership and post-clustering common break estimates are given in Table~\ref{tab:real_stock_common_break_groups}.

\begin{table}[htbp]
\centering
\caption{Estimated Latent Common-Break Groups in the A-Share Application}
\label{tab:real_stock_common_break_groups}
\setlength{\tabcolsep}{6pt}
\renewcommand{\arraystretch}{1.15}
\begin{tabular}{cccc}
\toprule
\textbf{Group} & $\widehat b_k$ & \textbf{Estimated Common Break Date} & \textbf{Stocks} \\
\midrule
1 & 98  & 2023-04-07 & 002415.SZ \\
2 & 281 & 2024-01-08 & 300750.SZ, 601888.SH, 600309.SH \\
3 & 349 & 2024-04-23 & 600048.SH, 601857.SH, 000002.SZ \\
4 & 493 & 2024-11-26 & 600000.SH \\
\bottomrule
\end{tabular}
\end{table}

The clustering step plays an important role in the empirical interpretation. It does not simply list stock-level break dates, but further determines whether the detected breaks are synchronized across stocks. In this application, Groups 2 and 3 contain multiple stocks, indicating that some selected stocks share similar structural break timing. Groups 1 and 4 are singleton groups, suggesting that the corresponding changes are more stock-specific.

The estimated groups are also economically interpretable. Group 1 contains 002415.SZ, Hikvision. Its estimated break date is April 7, 2023, a period when A-share technology and artificial-intelligence-related stocks were actively repriced. Hikvision has substantial exposure to AIoT, machine vision, and digital transformation applications. The detected break therefore appears consistent with the early-2023 revaluation of AI and digital-economy themes, together with firm-specific reassessment of Hikvision's growth profile.

Group 2 contains 300750.SZ, 601888.SH, and 600309.SH, with a pooled break estimate on January 8, 2024. This group combines new-energy, duty-free consumption, and chemical-sector stocks. The common break date coincides with a broad weakening of the A-share market at the beginning of 2024, when the Shanghai Composite Index fell below 2900 and most stocks declined. The group also reflects sector-specific pressures around late 2023 and early 2024. For 300750.SZ, CATL, the individual break occurs near December 18, 2023, when CATL fell sharply and the ChiNext index reached a new low. The abrupt termination of Germany's electric-vehicle subsidy programme in December 2023 also added pressure to the global electric-vehicle and battery supply chain. For 601888.SH, China Tourism Group Duty Free, the late-2023 adjustment is consistent with the market's reassessment of post-pandemic consumption recovery and duty-free demand. Thus, Group 2 can be interpreted as a common intraday regime shift associated with growth-stock valuation pressure and sector-level demand reassessment.

Group 3 contains 600048.SH, 601857.SH, and 000002.SZ, with a common break estimate on April 23, 2024. Two of the stocks, 600048.SH and 000002.SZ, are major real-estate developers, while 601857.SH is PetroChina, a large state-owned high-dividend energy stock. Around April 2024, Chinese real-estate equities were under substantial valuation pressure, while policy expectations for stabilizing the property market were intensifying. At the same time, the A-share market showed a stronger preference for state-owned, high-dividend, and defensive large-cap assets. The estimated group therefore appears to capture a broader style and policy-driven regime shift involving real-estate stress, policy support expectations, and high-dividend state-owned enterprise allocation.

Group 4 contains only 600000.SH, Shanghai Pudong Development Bank, with an estimated break date of November 26, 2024. Around this period, the banking sector showed visibly stronger trading activity, and several bank stocks approached stage highs. The break in 600000.SH is therefore consistent with a banking-sector repricing episode related to high-dividend allocation, expectations of improving asset quality, and defensive positioning in the equity market.

To better understand the nature of the estimated breaks, Figure~\ref{fig:real_stock_group_mean_prepost} plots the estimated pre- and post-break mean functions for each latent group. The curves show clear changes in the intraday cumulative return patterns after the estimated group-level break dates. For Groups 2, 3, and 4, the post-break mean curves lie above the pre-break curves over most intraday grid points, indicating a shift in the average within-day return profile. Group 1 displays the opposite pattern, with a substantially lower post-break mean curve. These differences indicate that the detected breaks correspond to meaningful changes in the average intraday return dynamics rather than isolated daily price movements.

\begin{figure}[htbp]
\centering
\includegraphics[width=0.92\textwidth]{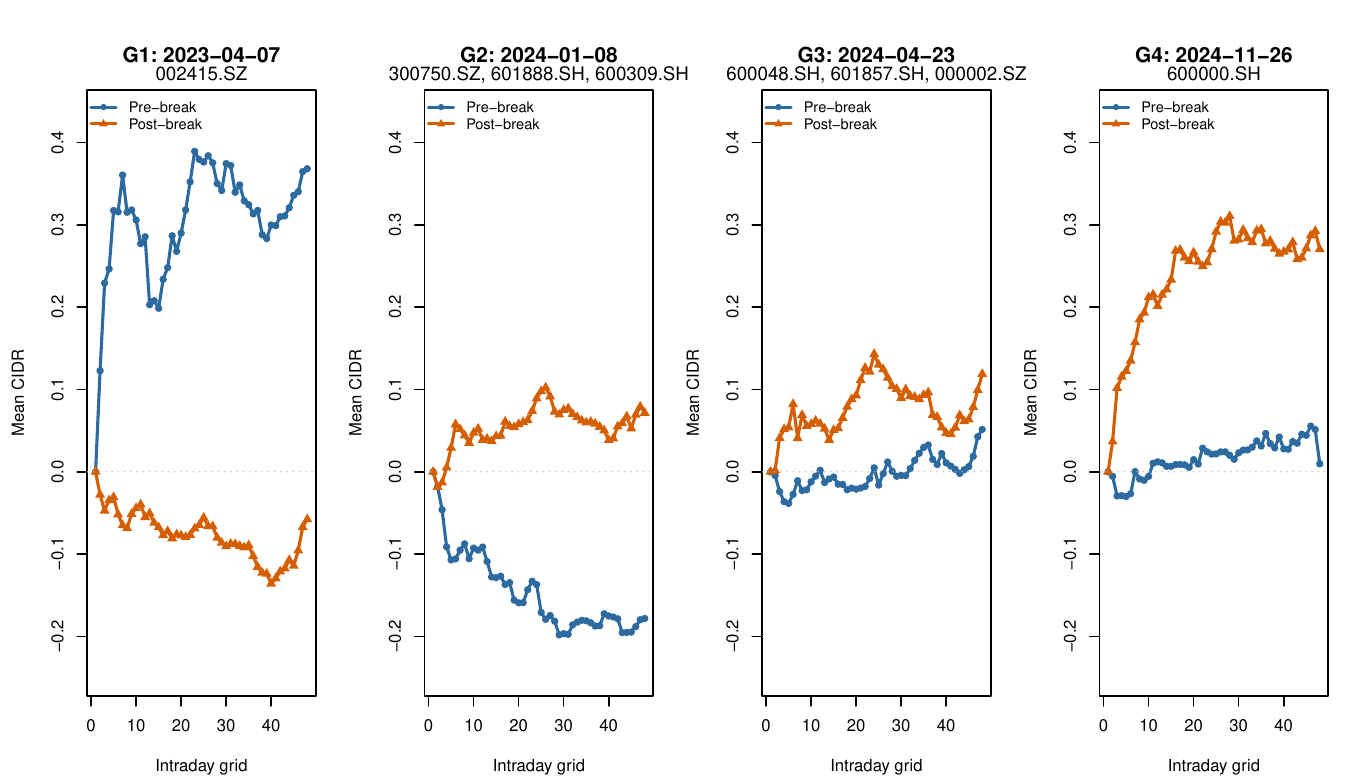}
\caption{Selected A-share stocks exhibit distinct pre- and post-break intraday mean functions. Curves show the estimated mean functions before and after the group-level breaks selected by Energy--PE with $c_\xi=0.10$. The horizontal axis denotes the five-minute intraday grid.}
\label{fig:real_stock_group_mean_prepost}
\end{figure}

The empirical findings provide marginal but economically meaningful evidence of structural changes in the A-share intraday functional panel. The global Energy--PE test rejects the null hypothesis at the 10\% level, and the subsequent post-test analysis identifies a sparse subset of stocks whose intraday functional mean patterns exhibit structural changes. The clustering step further provides an interpretable summary of the latent common-break structure. In particular, the method distinguishes stock-specific breaks from common-break groups and links the detected changes to economically meaningful market episodes, including AI-related repricing, new-energy and consumption-sector pressure, real-estate policy expectations, and banking-sector high-dividend allocation.

\section{Conclusion}
\label{sec:conclusion}

This paper develops an Energy--PE procedure for testing and estimating structural breaks in high-dimensional functional panel data. The key methodological step is to aggregate squared subject-specific CUSUM energies before adding a generalized power-enhancement component. This construction preserves signal strength under dense, moderately dense, and sign-canceling alternatives, while the enhancement term improves sensitivity to sparse strong alternatives. The counterexample in Section~\ref{sec:model} shows why linearly aggregated PE--CUSUM statistics can miss symmetric break patterns, and the proposed energy aggregation directly addresses this failure mode.

The asymptotic theory establishes bootstrap calibration under the null, consistency under dense and sparse alternatives, changed-subject screening, preliminary localization, latent break-group recovery, and post-clustering pooled refinement. The simulation design assesses the method across different functional-generation mechanisms, sparsity regimes, and jump-direction patterns, and the reported size results indicate that the proposed calibration is stable in the considered settings. The empirical application to intraday stock return curves further illustrates the practical value of combining global testing, subject screening, and latent common-break grouping in a single inferential workflow.

The proposed framework is designed for mean-shift alternatives in high-dimensional functional panels. Extensions to changes in covariance structure, fully data-driven tuning of the enhancement threshold, and sequential monitoring settings are natural directions for future work.


\appendix

\section{Sign-Cancellation Counterexample}
\label{app:A}

We provide the detailed construction behind the counterexample summarized in Section~\ref{sec:model}. Consider a panel of functional observations
\(\{X_{i,t}(u): i=1,\ldots,N,\ t=1,\ldots,T\}\) in \(L^2[0,1]\) generated by
\begin{equation}\label{eq:counterexample_model}
X_{i,t}(u)
=
\varepsilon_{i,t}(u)+\mathbf 1\{t>\tau\}\delta_i(u),
\qquad
\tau=\lfloor T/2\rfloor .
\end{equation}
The noise process admits the finite Karhunen--Lo\`eve representation
\[
\varepsilon_{i,t}(u)
=
\sum_{j=1}^{J}\eta_{i,t,j}\phi_j(u),
\]
where \(\{\phi_j\}\) are Fourier basis functions on \([0,1]\), and the coefficients \(\eta_{i,t,j}\) are centered Gaussian variables with uniformly bounded variances. The change functions are
\[
\delta_i(u)=s_i\,\delta\,\delta^\star(u),
\qquad
s_i\in\{+1,-1\},
\]
where exactly half of the subjects have \(s_i=+1\) and the remaining half have \(s_i=-1\). The common direction is
\[
\delta^\star(u)=\frac{1}{\sqrt m}\sum_{k=1}^{m}\phi_k(u),
\]
with fixed \(m\). Thus
\[
\frac1N\sum_{i=1}^{N}\delta_i(u)\equiv0,
\qquad
\frac1N\sum_{i=1}^{N}\|\delta_i\|_{L^2}^2
=
\delta^2\|\delta^\star\|_{L^2}^2>0.
\]
The first identity removes the deterministic signal from any statistic based on the cross-sectional average, whereas the second identity shows that subject-wise squared-energy aggregation retains the signal.

\begin{proposition}[Why PE--CUSUM fails but Energy--CUSUM succeeds]
\label{prop:short_counterexample}
Suppose that $N$ is even, $\tau_i=\tau_0=\lfloor Tr_0\rfloor$ for all $i$, and
\[
\delta_i(u)=\delta(u),\quad i=1,\dots,N/2,
\qquad
\delta_i(u)=-\delta(u),\quad i=N/2+1,\dots,N,
\]
for some $\delta\in L^2(C)\setminus\{0\}$. Then:

\begin{enumerate}
\item[(1)] the panel-average jump vanishes,
\[
\frac1N\sum_{i=1}^N\delta_i(u)=0,
\]
so the mean-based CUSUM component in PE--CUSUM does not contain any deterministic break signal;

\item[(2)] if
\[
T\|\delta\|_{L^2}^2=o(\xi_{NT}),
\]
and the subject-wise CUSUM noise energies obey a uniform upper-tail bound
\[
N\sup_{1\le i\le N}
P\!\left(
\sup_{0\le x\le1}\int_C \{C_{i,T}^{(\varepsilon)}(x,u)\}^2\,du>\xi_{NT}/4
\right)\to0,
\]
then
\[
P\!\left(Z_{NT}^{\diamond}=0\right)\to 1,
\]
and hence PE--CUSUM has no asymptotic power under this alternative;

\item[(3)] for the Energy--CUSUM statistic,
\[
E_T(r_0)
=
E_T^{(\varepsilon)}(r_0)
+
r_0^2(1-r_0)^2T\|\delta\|_{L^2}^2
+
o_P\!\bigl(T\|\delta\|_{L^2}^2\bigr),
\]
provided that
\[
T\|\delta\|_{L^2}^2\to\infty.
\]
Therefore, Energy--CUSUM remains powerful in this sign-cancellation regime.
\end{enumerate}
\end{proposition}

\begin{proof}[Proof]
We split the proof into three parts.

\medskip
\noindent
\textbf{Step 1: Failure of the mean-based CUSUM component.}
Let
\[
\bar X_t(u)=\frac{1}{N}\sum_{i=1}^N X_{i,t}(u).
\]
Under the symmetric construction,
\[
\frac{1}{N}\sum_{i=1}^N \delta_i(u)
=
\frac{1}{N}\left(\frac{N}{2}\delta(u)-\frac{N}{2}\delta(u)\right)
=0.
\]
Hence
\[
\bar X_t(u)=\bar\varepsilon_t(u),
\]
so the panel-average process contains no break signal at all. Therefore the CUSUM statistic $Z_{NT}$ constructed from $\bar X_t$ has the same deterministic part as under $H_0$, and thus cannot accumulate signal under this alternative.

\medskip
\noindent
\textbf{Step 2: Failure of the PE component.}
For each subject $i$, the subject-wise functional CUSUM process admits the decomposition
\[
C_{i,T}(x,u)
=
C_{i,T}^{(\varepsilon)}(x,u)+C_{i,T}^{(\delta)}(x,u),
\]
where the signal part reaches its maximal order at $x=r_0$ and satisfies
\[
\sup_{0\le x\le 1}\int_C \{C_{i,T}^{(\delta)}(x,u)\}^2\,du
\asymp
T\|\delta\|_{L^2}^2.
\]
By assumption,
\[
T\|\delta\|_{L^2}^2=o(\xi_{NT}).
\]
Since the deterministic signal part is $o(\xi_{NT})$ and the noise energies satisfy the stated uniform upper-tail bound, a union bound gives
\[
N\sup_{1\le i\le N}P\!\left(
\sup_{0\le x\le 1}\int_C C_{i,T}^2(x,u)\,du>\xi_{NT}
\right)\to 0.
\]
Therefore
\[
P(Z_{NT}^{\diamond}=0)\to 1.
\]
Therefore,
\[
\widehat Z_{NT}=Z_{NT}+Z_{NT}^{\diamond}
=
Z_{NT}+o_P(1),
\]
and PE--CUSUM has no asymptotic power here.

\medskip
\noindent
\textbf{Step 3: Success of Energy--CUSUM.}
At the true break fraction $x=r_0$, write
\[
C_{i,T}(r_0,u)=C_{i,T}^{(\varepsilon)}(r_0,u)+s_i\,\sqrt{T}\,r_0(1-r_0)\delta(u),
\]
where $s_i=1$ for $i\le N/2$ and $s_i=-1$ otherwise. Squaring and integrating give
\[
\int_C C_{i,T}^2(r_0,u)\,du
=
\int_C \{C_{i,T}^{(\varepsilon)}(r_0,u)\}^2\,du
+
T r_0^2(1-r_0)^2\|\delta\|_{L^2}^2
+
2s_i\sqrt{T}r_0(1-r_0)\langle C_{i,T}^{(\varepsilon)}(r_0,\cdot),\delta\rangle.
\]
Summing over $i$ yields
\[
E_T(r_0)
=
E_T^{(\varepsilon)}(r_0)
+
r_0^2(1-r_0)^2T\|\delta\|_{L^2}^2
+
\widetilde R_{NT},
\]
where the averaged cross term $\widetilde R_{NT}$ is centered and, by independence across $i$,
\[
\widetilde R_{NT}=O_P(\sqrt{T/N}\,\|\delta\|_{L^2}).
\]
Thus the normalized deterministic signal term is of order $T\|\delta\|_{L^2}^2$, whereas the stochastic remainder is of smaller order whenever
\[
T\|\delta\|_{L^2}^2\to\infty.
\]
Hence the Energy--CUSUM aggregation accumulates the aligned weak signals across subjects after squaring, and does not suffer from the sign-cancellation effect.
\end{proof}

\section{Proofs of Main Results}
\label{app:B}

This appendix provides proofs of the main theorems in Section~\ref{sec:theory}.
The proofs rely on several auxiliary propositions stated and proved below, together with technical lemmas collected in Appendix~\ref{app:C}.

Throughout, let
\[
C_{i,T}(x,u)
=
\frac{1}{\sqrt{T}}
\left[
\sum_{t=1}^{\lfloor Tx\rfloor} X_{it}(u)
-
\frac{\lfloor Tx\rfloor}{T}\sum_{t=1}^{T} X_{it}(u)
\right],
\qquad 0\le x\le 1,
\]
and define
\[
E_T(x)
=
\frac{1}{N}\sum_{i=1}^{N}\int_C C_{i,T}^2(x,u)\,du,
\qquad
Z_T^{\mathrm{Energy}}
=
\sup_{0\le x\le 1}E_T(x).
\]
Recall also that
\[
Y_{iT}
=
\sup_{0\le x\le 1}\int_C C_{i,T}^2(x,u)\,du,
\]
and
\[
\mathcal Z_{NT}^{g}
=
r_{NT}\sum_{i=1}^{N}
g\!\left(\frac{Y_{iT}-\xi_{NT}}{\tau_{NT}}\right),
\]
where $g:\mathbb R\to[0,\infty)$ is the generalized enhancement function introduced in Section~\ref{sec:model}. The unified statistic is
\[
\widehat Z_{NT}^{\mathrm{EPE}}
=
Z_T^{\mathrm{Energy}}+\mathcal Z_{NT}^{g}.
\]

For later use, define the error-only CUSUM process
\[
C_{i,T}^{\varepsilon}(x,u)
=
\frac{1}{\sqrt{T}}
\left[
\sum_{t=1}^{\lfloor Tx\rfloor} \varepsilon_{it}(u)
-
\frac{\lfloor Tx\rfloor}{T}\sum_{t=1}^{T}\varepsilon_{it}(u)
\right],
\]
and the deterministic signal part
\[
D_{i,T}(x,u)
=
\frac{1}{\sqrt{T}}
\left[
\sum_{t=1}^{\lfloor Tx\rfloor} \delta_i(u)\mathbf 1\{t>\tau_i\}
-
\frac{\lfloor Tx\rfloor}{T}\sum_{t=1}^{T}\delta_i(u)\mathbf 1\{t>\tau_i\}
\right].
\]
Then
\begin{equation}
C_{i,T}(x,u)=C_{i,T}^{\varepsilon}(x,u)+D_{i,T}(x,u).
\label{eq:A.1}
\end{equation}

For clarity and self-containment, we first isolate the auxiliary propositions that drive the main theoretical arguments. Each proposition is stated and proved before the theorem proofs, so that the subsequent arguments can refer to these intermediate results without interrupting the logical flow.

\begin{proposition}
\label{prop:A.1}
Suppose that the conditions of Theorem \ref{thm:3.1} hold.
Then the Energy--CUSUM statistic $Z_T^{\mathrm{Energy}}$ satisfies
\[
\sup_{z\in\mathbb R}
\left|
P_{H_0}\!\left(Z_T^{\mathrm{Energy}}\le z\right)
-
P\!\left(Z_{NT}^{G}\le z\right)
\right|\to0.
\]
Moreover, the rejection rule based on $c_{NT,\alpha}^{*}$ has asymptotic size at most $\alpha$.
\end{proposition}

\begin{proof}[Proof of Proposition \ref{prop:A.1}]
The Gaussian approximation follows from Lemma \ref{lem:B.1}. The size statement follows by combining this approximation with the bootstrap consistency condition \eqref{eq:3.bootstrap}.
\end{proof}

\begin{proposition}
\label{prop:A.2}
Suppose that Assumptions \ref{ass:regularity} and \ref{ass:dense} hold.
Then for $x_0=r_0$ given in \eqref{eq:3.B4},
\[
E_T(x_0)\xrightarrow{P}\infty.
\]
Consequently,
\[
Z_T^{\mathrm{Energy}}\xrightarrow{P}\infty.
\]
\end{proposition}

\begin{proof}[Proof of Proposition \ref{prop:A.2}]
The proposition follows immediately from \eqref{eq:A.2}, Lemma \ref{lem:B.2},
Lemma \ref{lem:B.3}, and condition \eqref{eq:3.B5}.
\end{proof}

\begin{proposition}
\label{prop:A.3}
Suppose that Assumption \ref{ass:enhancement}(1) holds. Under $H_0$,
\[
\mathcal Z_{NT}^{g}=o_P(1).
\]
\end{proposition}

\begin{proof}[Proof of Proposition \ref{prop:A.3}]
Under $H_0$,
\[
\mathcal Z_{NT}^{g}
=
r_{NT}\sum_{i=1}^{N}
g\!\left(\frac{Y_{iT}-\xi_{NT}}{\tau_{NT}}\right).
\]
Since each summand is nonnegative,
\[
E\bigl[\mathcal Z_{NT}^{g}\mid H_0\bigr]
=
r_{NT}\sum_{i=1}^{N}
E\!\left[
g\!\left(\frac{Y_{iT}-\xi_{NT}}{\tau_{NT}}\right)\middle|H_0
\right]
\le
Nr_{NT}\sup_{1\le i\le N}
E\!\left[
g\!\left(\frac{Y_{iT}-\xi_{NT}}{\tau_{NT}}\right)\middle|H_0
\right].
\]
By condition \eqref{eq:3.B6}, the right-hand side converges to zero. Therefore,
\[
E\bigl[\mathcal Z_{NT}^{g}\mid H_0\bigr]\to0.
\]
Applying Markov's inequality yields
\[
\mathcal Z_{NT}^{g}=o_P(1).
\]
\end{proof}

\begin{proposition}
\label{prop:A.4}
Suppose that the sparse-power condition in Assumption \ref{ass:enhancement}(2) holds.
Then under the sparse strong alternative,
\[
\mathcal Z_{NT}^{g}\xrightarrow{P}\infty.
\]
\end{proposition}

\begin{proof}[Proof of Proposition \ref{prop:A.4}]
For $i\in J_N$, condition \eqref{eq:3.B7} implies
\[
P\!\left(
Y_{iT}\ge \xi_{NT}+\tau_{NT}b_{NT}
\right)\to1.
\]
Let
\[
A_{iT}=\{Y_{iT}\ge \xi_{NT}+\tau_{NT}b_{NT}\}.
\]
The uniform lower bound in \eqref{eq:3.B7} implies
\[
\frac{1}{|J_N|}\sum_{i\in J_N}\mathbf 1(A_{iT})=1-o_P(1)
\]
by Markov's inequality applied to the average number of failures.
Since $g$ is nondecreasing and nonnegative, on the event
\[
Y_{iT}\ge \xi_{NT}+\tau_{NT}b_{NT},
\]
we have
\[
g\!\left(\frac{Y_{iT}-\xi_{NT}}{\tau_{NT}}\right)\ge g(b_{NT}).
\]
Hence
\[
\mathcal Z_{NT}^{g}
\ge
r_{NT}\sum_{i\in J_N}
g\!\left(\frac{Y_{iT}-\xi_{NT}}{\tau_{NT}}\right)
\ge
(1-o_P(1))r_{NT}|J_N|\,g(b_{NT}).
\]
By \eqref{eq:3.B8}, the right-hand side diverges to infinity in probability.
Thus
\[
\mathcal Z_{NT}^{g}\xrightarrow{P}\infty.
\]
\end{proof}

\begin{proposition}
\label{prop:A.5}
Suppose that Assumption \ref{ass:regularity} and Assumption \ref{ass:screening}(1)--(3) hold.
Then
\[
P\!\left(
\widehat{\mathcal C}_0=\mathcal C_0,\
\widehat{\mathcal C}_1=\mathcal C_1
\right)\to1.
\]
Moreover, if Assumption \ref{ass:screening}(4) and Assumption \ref{ass:localmax} hold, then for any arbitrarily small $\zeta>0$,
\[
\max_{i\in\mathcal C_1}
|\widehat\tau_i-\tau_i|
=
o_P\!\left([\log(N\vee T)]^{1+\zeta}\right).
\]
\end{proposition}

\begin{proof}[Proof of Proposition \ref{prop:A.5}]
For unchanged subjects, condition \eqref{eq:3.9a} and the union bound imply
\[
P\!\left(\max_{i\in\mathcal C_0}Y_{iT}>\xi_{NT}\right)
\le
N\sup_{i\in\mathcal C_0}P(Y_{iT}>\xi_{NT})
\to0.
\]
For changed subjects, Lemma \ref{lem:B.4} gives
\[
P\!\left(\min_{i\in\mathcal C_1}Y_{iT}>\xi_{NT}\right)\to1.
\]
Thus $\widehat{\mathcal C}_0=\mathcal C_0$ and $\widehat{\mathcal C}_1=\mathcal C_1$ with probability tending to one.
For the break-point rate, let
\[
b_{NT}=[\log(N\vee T)]^{1+\zeta}.
\]
By Lemma \ref{lem:B.5}, the stochastic fluctuation of the squared CUSUM criterion is
$O_P(\log(N\vee T))=o_P(b_{NT})$ uniformly over $i$ and $t$.
By Lemma \ref{lem:B.6}, whenever $|t-\tau_i|\ge \epsilon b_{NT}$, the deterministic loss is at least
$c\epsilon b_{NT}$ uniformly over $i\in\mathcal C_1$.
Thus the deterministic loss dominates the stochastic fluctuation uniformly outside the window
$|t-\tau_i|<\epsilon b_{NT}$.
Since $\epsilon>0$ is arbitrary, the argmax comparison argument yields
\[
\max_{i\in\mathcal C_1}|\widehat\tau_i-\tau_i|
=
o_P(b_{NT}),
\]
which proves the asserted rate.
\end{proof}

\begin{proposition}
\label{prop:A.6}
Suppose that the latent structure \eqref{eq:3.14} and Assumptions
\ref{ass:regularity}, \ref{ass:screening}, \ref{ass:localmax}, and \ref{ass:cluster} hold.
Then
\[
P(\widehat K=K_0)\to1,
\]
and
\[
P\!\left(
\widehat{\mathcal C}(b_k)=\mathcal C(b_k),\ k=1,\ldots,K_0
\ \Big|\
\widehat K=K_0
\right)\to1.
\]
\end{proposition}

\begin{proof}[Proof of Proposition \ref{prop:A.6}]
By Lemma \ref{lem:B.7}, the ordered preliminary estimates display $K_0$ well-separated blocks with probability tending to one.
Hence, if $K=K_0$, the gap-based construction recovers the correct partition.

For $K<K_0$, Lemma \ref{lem:B.8} implies
\[
IC(K)-IC(K_0)
=
\log V(K)-\log V(K_0)-(K_0-K)\rho_{NT}
>0
\]
with probability tending to one, because the positive loss inflation dominates the smaller penalty.

For $K>K_0$, Lemma \ref{lem:B.9} gives
\[
IC(K)-IC(K_0)
=
\log V(K)-\log V(K_0)+(K-K_0)\rho_{NT}
>0
\]
with probability tending to one, because the penalty dominates the negligible fit gain.

Therefore $\widehat K=K_0$ with probability tending to one.
Conditional on this event, the block separation argument yields the desired membership consistency.
\end{proof}

\begin{proposition}
\label{prop:A.7}
Suppose that the conditions of Theorem \ref{thm:3.7} hold.
Then
\[
\max_{1\le k\le K_0}|\widetilde b_k-b_k|=o_P(T).
\]
\end{proposition}

\begin{proof}[Proof of Proposition \ref{prop:A.7}]
On the event of correct cluster recovery, the pooled estimator $\widetilde b_k$ maximizes the criterion in Lemma \ref{lem:B.10}.
For any fixed $\epsilon>0$, that lemma implies that, with probability tending to one, the criterion at $t=b_k$ strictly exceeds the supremum of the criterion over all $t$ satisfying $|t-b_k|>\epsilon T$.
Hence $|\widetilde b_k-b_k|\le \epsilon T$ with probability tending to one for each $k$.
Since $\epsilon$ is arbitrary and $K_0$ is fixed, taking the intersection over $k=1,\ldots,K_0$ yields
\[
\max_{1\le k\le K_0}|\widetilde b_k-b_k|=o_P(T),
\]
as claimed.
\end{proof}

\subsection{Proof of Theorem \ref{thm:3.1}}
\label{subsec:proof-thm-3-1}

\begin{proof}[Proof of Theorem \ref{thm:3.1}]
Under $H_0$, the model reduces to
\[
X_{it}(u)=\mu_i(u)+\varepsilon_{it}(u),
\]
where the mean function $\mu_i$ is time-invariant.
Hence the deterministic mean part is annihilated by the CUSUM transformation, and
\[
C_{i,T}(x,u)=C_{i,T}^{\varepsilon}(x,u).
\]

The Gaussian approximation is exactly \eqref{eq:3.GA}. Combining \eqref{eq:3.GA} with the bootstrap consistency condition \eqref{eq:3.bootstrap}, and using the anti-concentration condition to transfer the Gaussian quantile to the bootstrap critical value, yields the stated asymptotic size control for the bootstrap rejection rule.
\end{proof}

\subsection{Proof of Theorem \ref{thm:3.2}}
\label{subsec:proof-thm-3-2}

\begin{proof}[Proof of Theorem \ref{thm:3.2}]
Fix $x_0=r_0$, where $r_0$ is specified in Assumption \ref{ass:dense}.
By \eqref{eq:A.1},
\[
C_{i,T}(x_0,u)=C_{i,T}^{\varepsilon}(x_0,u)+D_{i,T}(x_0,u).
\]
Using the elementary inequality in a Hilbert space,
\[
\|a+b\|_{L^2(C)}^2
\ge \frac12\|a\|_{L^2(C)}^2-\|b\|_{L^2(C)}^2,
\]
we obtain
\begin{align}
E_T(x_0)
&=
\frac1N\sum_{i=1}^N\|C_{i,T}(x_0)\|_{L^2(C)}^2 \notag\\
&\ge
\frac1N\sum_{i\in I_N}\|C_{i,T}(x_0)\|_{L^2(C)}^2 \notag\\
&\ge
\frac{1}{2N}\sum_{i\in I_N}\|D_{i,T}(x_0)\|_{L^2(C)}^2
-
\frac1N\sum_{i\in I_N}\|C_{i,T}^{\varepsilon}(x_0)\|_{L^2(C)}^2.
\label{eq:A.2}
\end{align}

By Lemma \ref{lem:B.3}, under the alignment condition \eqref{eq:3.B4},
there exists a constant $c_0>0$ such that, uniformly for $i\in I_N$,
\[
\|D_{i,T}(x_0)\|_{L^2(C)}^2
\ge
c_0 T\|\delta_i\|_{L^2(C)}^2
\]
for all sufficiently large $T$.
Hence the first term on the right-hand side of \eqref{eq:A.2} is bounded below by
\[
\frac{c_0T}{2N}\sum_{i\in I_N}\|\delta_i\|_{L^2(C)}^2.
\]
By condition \eqref{eq:3.B5}, this term diverges to $+\infty$.

On the other hand, by Lemma \ref{lem:B.2},
\[
\frac1N\sum_{i\in I_N}\|C_{i,T}^{\varepsilon}(x_0)\|_{L^2(C)}^2=O_P(1).
\]
Therefore,
\[
E_T(x_0)\xrightarrow{P}\infty.
\]
Since
\[
Z_T^{\mathrm{Energy}}=\sup_{0\le x\le1}E_T(x)\ge E_T(x_0),
\]
it follows that
\[
Z_T^{\mathrm{Energy}}\xrightarrow{P}\infty.
\]
This proves \eqref{eq:3.5}.
\end{proof}

\subsection{Proof of Theorem \ref{thm:3.3}}
\label{subsec:proof-thm-3-3}

\begin{proof}[Proof of Theorem \ref{thm:3.3}]
By definition,
\[
\widehat Z_{NT}^{\mathrm{EPE}}
=
Z_T^{\mathrm{Energy}}+\mathcal Z_{NT}^{g}.
\]
Under $H_0$, Proposition \ref{prop:A.3} yields
\[
\mathcal Z_{NT}^{g}=o_P(1).
\]
Hence
\begin{equation}
\widehat Z_{NT}^{\mathrm{EPE}}
=
Z_T^{\mathrm{Energy}}+o_P(1).
\label{eq:A.3}
\end{equation}
Combining \eqref{eq:A.3} with the Gaussian approximation and bootstrap size statement in Theorem \ref{thm:3.1} proves \eqref{eq:3.7}.
The anti-concentration condition in Theorem \ref{thm:3.1} also makes the rejection probability stable under the \(o_P(1)\) perturbation contributed by \(\mathcal Z_{NT}^{g}\), which gives the asserted asymptotic size control for the full Energy--PE statistic.
\end{proof}

\subsection{Proof of Theorem \ref{thm:3.4}}
\label{subsec:proof-thm-3-4}

\begin{proof}[Proof of Theorem \ref{thm:3.4}]
Since $c_{NT,\alpha}^{*}=O_P(1)$ by assumption, it suffices to show that
\[
\widehat Z_{NT}^{\mathrm{EPE}}\xrightarrow{P}\infty
\]
under either condition (1) or condition (2).

\medskip
\noindent
\textbf{Case 1: dense, aligned, or sign-canceling alternatives.}
Suppose that condition (1) holds.
By Theorem \ref{thm:3.2},
\[
Z_T^{\mathrm{Energy}}\xrightarrow{P}\infty.
\]
Since $\mathcal Z_{NT}^{g}\ge0$,
\[
\widehat Z_{NT}^{\mathrm{EPE}}
=
Z_T^{\mathrm{Energy}}+\mathcal Z_{NT}^{g}
\ge
Z_T^{\mathrm{Energy}}.
\]
Hence
\[
\widehat Z_{NT}^{\mathrm{EPE}}\xrightarrow{P}\infty.
\]

\medskip
\noindent
\textbf{Case 2: sparse strong alternatives.}
Suppose that condition (2) holds.
By Proposition \ref{prop:A.4},
\[
\mathcal Z_{NT}^{g}\xrightarrow{P}\infty.
\]
Thus
\[
\widehat Z_{NT}^{\mathrm{EPE}}
\ge
\mathcal Z_{NT}^{g}.
\]
It follows that
\[
\widehat Z_{NT}^{\mathrm{EPE}}\xrightarrow{P}\infty.
\]

Therefore, in either case,
\[
P\!\left(\widehat Z_{NT}^{\mathrm{EPE}}>c_{NT,\alpha}^{*}\right)\to1.
\]
This proves \eqref{eq:3.8}.
\end{proof}

\subsection{Proof of Theorem \ref{thm:3.5}}
\label{subsec:proof-thm-3-5}

\begin{proof}[Proof of Theorem \ref{thm:3.5}]
The theorem consists of two parts: screening consistency and uniform accuracy of the preliminary break-point estimators.

\medskip
\noindent
\textbf{Step 1: screening consistency.}
For each $i\in\mathcal C_0$, under $H_0$ we have $\delta_i\equiv0$. By the false-positive control in Assumption \ref{ass:screening}(1),
\[
N\sup_{i\in\mathcal C_0}P(Y_{iT}>\xi_{NT})\to0.
\]
Therefore, by the union bound,
\[
P(\exists\, i\in\mathcal C_0: i\in\widehat{\mathcal C}_1)\to0.
\]
Thus no false positives occur with probability tending to one.

For $i\in\mathcal C_1$, Assumption \ref{ass:screening}(2)--(3) implies that the observed subject-wise energy dominates the threshold:
\[
\frac{T\omega_{Ti}^2\|\delta_i\|_{L^2(C)}^2}{\xi_{NT}}\to\infty
\]
uniformly over $i\in\mathcal C_1$.
By Lemma \ref{lem:B.4}, this yields
\[
P\!\left(\min_{i\in\mathcal C_1}Y_{iT}>\xi_{NT}\right)\to1
\].
Hence
\[
P(\exists\, i\in\mathcal C_1: i\notin\widehat{\mathcal C}_1)\to0.
\]
Combining the two parts gives
\[
P\!\left(
\widehat{\mathcal C}_0=\mathcal C_0,\
\widehat{\mathcal C}_1=\mathcal C_1
\right)\to1.
\]

\medskip
\noindent
\textbf{Step 2: preliminary break-point estimation.}
Conditional on the event $\widehat{\mathcal C}_1=\mathcal C_1$, for each $i\in\mathcal C_1$,
the estimator $\widehat\tau_i$ maximizes the individual quadratic CUSUM objective
\[
Q_{i,T}(t)
:=
\int_C C_{i,T}^2(t/T,u)\,du.
\]
By the standard localization argument for functional change-point estimators,
the deterministic part of $Q_{i,T}(t)$ is uniquely maximized at $t=\tau_i$,
while the stochastic remainder is uniformly of smaller order.
Under Assumption \ref{ass:screening}(4) and Assumption \ref{ass:localmax}, Lemmas \ref{lem:B.5}--\ref{lem:B.6}
show that the argmax error obeys
\[
\max_{i\in\mathcal C_1}
|\widehat\tau_i-\tau_i|
=
o_P\!\left([\log(N\vee T)]^{1+\zeta}\right).
\]
This proves \eqref{eq:3.13}.
\end{proof}

\subsection{Proof of Theorem \ref{thm:3.6}}
\label{subsec:proof-thm-3-6}

\begin{proof}[Proof of Theorem \ref{thm:3.6}]
By Theorem \ref{thm:3.5}, with probability tending to one,
\[
\widehat{\mathcal C}_1=\mathcal C_1
\]
and
\[
\max_{i\in\mathcal C_1}|\widehat\tau_i-\tau_i|
=
o_P\!\left([\log(N\vee T)]^{1+\zeta}\right).
\]
Hence the preliminary break-point estimates cluster tightly around the true distinct break points
$\{b_1,\ldots,b_{K_0}\}$.

Under Assumption \ref{ass:cluster}(1), the true break points are separated at order $T$:
\[
|b_{k_1}-b_{k_2}|\asymp T,\qquad k_1\neq k_2.
\]
Since the estimation error is only logarithmic in $N\vee T$, it follows that the within-group dispersion
of the preliminary estimates is asymptotically negligible relative to the between-group gaps.
Therefore, the $(K_0-1)$ largest empirical gaps in the ordered sequence
$\widehat\tau_{(1)},\ldots,\widehat\tau_{(n)}$ coincide with the true group boundaries
with probability tending to one.

This establishes that, if $K=K_0$, then
\[
P\!\left(
\widehat{\mathcal C}(k\mid K_0)=\mathcal C(b_k),\ k=1,\ldots,K_0
\right)\to1.
\]

It remains to show that the information criterion consistently selects $K_0$.
Let
\[
IC(K)=\log V(K)+K\rho_{NT}.
\]
For $K<K_0$, at least one true cluster is incorrectly merged with another, which yields a non-negligible increase
in the within-cluster fitting loss $V(K)$.
By Lemma \ref{lem:B.8}, this loss dominates the penalty reduction from using fewer groups, and thus
\[
P\bigl(IC(K)>IC(K_0)\bigr)\to1,\qquad K<K_0.
\]
For $K>K_0$, some true cluster must be spuriously split.
In this case the reduction in $\log V(K)$ is asymptotically negligible, while the additional penalty
$(K-K_0)\rho_{NT}$ is positive and dominates by Assumption \ref{ass:cluster}(4)--(5).
Hence
\[
P\bigl(IC(K)>IC(K_0)\bigr)\to1,\qquad K>K_0.
\]
Therefore,
\[
P(\widehat K=K_0)\to1.
\]
Conditional on $\widehat K=K_0$, the membership consistency already established above implies
\[
P\!\left(
\widehat{\mathcal C}(b_k)=\mathcal C(b_k),\ k=1,\ldots,K_0
\ \Big|\
\widehat K=K_0
\right)\to1.
\]
This proves \eqref{eq:3.19} and \eqref{eq:3.20}.
\end{proof}

\subsection{Proof of Theorem \ref{thm:3.7}}
\label{subsec:proof-thm-3-7}

\begin{proof}[Proof of Theorem \ref{thm:3.7}]
By Theorem \ref{thm:3.6}, we may work on an event whose probability tends to one such that
\[
\widehat{\mathcal C}(b_k)=\mathcal C(b_k),\qquad k=1,\ldots,K_0.
\]
Thus, for each $k$,
\[
\widetilde b_k
=
\arg\max_{1\le t\le T}
\sum_{i\in\mathcal C(b_k)}
\int_C C_{i,T}^2(t/T,u)\,du.
\]
This is a pooled version of the individual break-point criterion.

Because all subjects in $\mathcal C(b_k)$ share the same true break point $b_k$,
pooling amplifies the deterministic signal component while averaging out stochastic fluctuations.
Under condition \eqref{eq:3.22}, the pooled signal strength diverges fast enough to dominate the noise uniformly outside any shrinking neighborhood of $b_k$.
More precisely, Lemma \ref{lem:B.10} implies that for every fixed $\epsilon>0$,
\[
P\!\left(
\sup_{|t-b_k|>\epsilon T}
\sum_{i\in\mathcal C(b_k)}
\int_C C_{i,T}^2(t/T,u)\,du
<
\sum_{i\in\mathcal C(b_k)}
\int_C C_{i,T}^2(b_k/T,u)\,du
\right)\to1.
\]
Hence $|\widetilde b_k-b_k|\le \epsilon T$ with probability tending to one.
Since $\epsilon>0$ is arbitrary,
\[
\max_{1\le k\le K_0}|\widetilde b_k-b_k|=o_P(T).
\]
This proves \eqref{eq:3.23}.
\end{proof}


\section{Technical Lemmas}
\label{app:C}

This appendix collects the technical lemmas used in Appendix~\ref{app:B}.

\subsection{Gaussian Approximation for Energy--CUSUM Statistics}
\label{subsec:gaussian-approximation}

\begin{lemma}
\label{lem:B.1}
Suppose that the Gaussian approximation condition \eqref{eq:3.GA} holds. Then
\[
\sup_{z\in\mathbb R}
\left|
P_{H_0}\!\left(Z_T^{\mathrm{Energy}}\le z\right)
-
P\!\left(Z_{NT}^{G}\le z\right)
\right|\to0.
\]
\end{lemma}

\begin{proof}
This is exactly condition \eqref{eq:3.GA}, which is stated at the theorem level because it is used only for null calibration.
\end{proof}

\begin{lemma}
\label{lem:B.2}
Suppose that Assumption \ref{ass:regularity}(1)--(2) holds.
Then, for each fixed $x_0\in(0,1)$,
\[
\frac1N\sum_{i\in I_N}\|C_{i,T}^{\varepsilon}(x_0)\|_{L^2(C)}^2=O_P(1).
\]
\end{lemma}

\begin{proof}
By Assumption \ref{ass:regularity}(1), each $C_{i,T}^{\varepsilon}(x_0)$ is a centered Hilbert-space random element
with uniformly bounded second moment.
Therefore,
\[
E\|C_{i,T}^{\varepsilon}(x_0)\|_{L^2(C)}^2\le C
\]
uniformly in $i$ and $T$.
Averaging over $i\in I_N$ and applying Markov's inequality yields the claim.
\end{proof}

\subsection{Signal Lower Bound for Energy--CUSUM}
\label{subsec:signal-lower-bound}

\begin{lemma}
\label{lem:B.3}
Suppose that Assumption \ref{ass:dense} holds.
Let $x_0=r_0$, where $r_0$ is the common alignment point in \eqref{eq:3.B4}.
Then there exists a constant $c_0>0$ such that, for all sufficiently large $T$,
\[
\|D_{i,T}(x_0)\|_{L^2(C)}^2
\ge
c_0 T\|\delta_i\|_{L^2(C)}^2
\]
uniformly over $i\in I_N$.
\end{lemma}

\begin{proof}
By direct calculation,
\[
D_{i,T}(x,u)
=
\sqrt{T}\,g_{i,T}(x)\delta_i(u),
\]
where
\[
g_{i,T}(x)
=
\frac{1}{T}
\left[
\sum_{t=1}^{\lfloor Tx\rfloor}\mathbf 1\{t>\tau_i\}
-
\frac{\lfloor Tx\rfloor}{T}\sum_{t=1}^{T}\mathbf 1\{t>\tau_i\}
\right].
\]
Hence
\[
\|D_{i,T}(x)\|_{L^2(C)}^2
=
T\,g_{i,T}^2(x)\,\|\delta_i\|_{L^2(C)}^2.
\]
Under \eqref{eq:3.B4}, $\tau_i/T\to r_0$ uniformly over $i\in I_N$.
Evaluating at $x=x_0=r_0$, the coefficient $g_{i,T}(x_0)$ converges uniformly to a nonzero constant
depending only on $r_0$.
Therefore, for all large $T$,
\[
g_{i,T}^2(x_0)\ge c_0
\]
for some $c_0>0$, uniformly in $i\in I_N$.
The claim follows.
\end{proof}

\subsection{Properties of the Generalized PE Component}
\label{subsec:generalized-pe-properties}

\begin{lemma}
\label{lem:B.4}
Suppose that Assumption \ref{ass:screening}(2)--(3) holds.
Then
\[
P\!\left(\min_{i\in\mathcal C_1}Y_{iT}>\xi_{NT}\right)\to1.
\]
\end{lemma}

\begin{proof}
By the same deterministic CUSUM calculation as in Lemma \ref{lem:B.3},
\[
Y_{iT}^{(0)}
=
\sup_{0\le x\le1}\int_C D_{i,T}^2(x,u)\,du
\asymp
T\omega_{Ti}^2\|\delta_i\|_{L^2(C)}^2.
\]
Assumption \ref{ass:screening}(3) gives
\[
\min_{i\in\mathcal C_1}
\frac{Y_{iT}}{T\omega_{Ti}^2\|\delta_i\|_{L^2(C)}^2}
\ge c+o_P(1)
\]
for some constant \(c>0\). Together with \eqref{eq:3.10}, this implies that the observed energy dominates the screening threshold uniformly over changed subjects. Hence
\[
P\!\left(\min_{i\in\mathcal C_1}Y_{iT}>\xi_{NT}\right)\to1.
\]
\end{proof}

\subsection{Uniform Break-Point Localization}
\label{subsec:uniform-localization}

\begin{lemma}
\label{lem:B.5}
Suppose that Assumption \ref{ass:localmax} holds.
Then
\[
\max_{1\le i\le N}\sup_{1\le t\le T}
\left|
\int_C C_{i,T}^2(t/T,u)\,du
-
E\!\left[\int_C C_{i,T}^2(t/T,u)\,du\right]
\right|
=
O_P\!\left(\log(N\vee T)\right).
\]
\end{lemma}

\begin{proof}
This is exactly the maximal inequality stated in Assumption \ref{ass:localmax}.
The assumption is separated from the baseline regularity conditions because it is needed only for the refined localization rate.
\end{proof}

\begin{lemma}
\label{lem:B.6}
Suppose that \eqref{eq:3.11} holds.
Then, for each $i\in\mathcal C_1$, the deterministic criterion
\[
Q_{i,T}^{(0)}(t)
:=
E\!\left[\int_C C_{i,T}^2(t/T,u)\,du\right]
\]
is uniquely maximized at $t=\tau_i$, and satisfies the local curvature bound
\[
Q_{i,T}^{(0)}(\tau_i)-Q_{i,T}^{(0)}(t)
\ge
c\,|t-\tau_i|
\]
for some constant $c>0$ uniformly in $i\in\mathcal C_1$ and $t$ in a neighborhood of $\tau_i$.
\end{lemma}

\begin{proof}
The deterministic objective is the squared norm of the CUSUM signal induced by a single mean shift.
Its explicit form is piecewise quadratic in $t/T$, with a unique maximum at the true break point.
Under the signal separation condition \eqref{eq:3.11}, the slope away from the maximizer is uniformly bounded below.
\end{proof}

\subsection{Information Criterion and Cluster Recovery}
\label{subsec:ic-cluster-recovery}

\begin{lemma}
\label{lem:B.7}
Suppose that the conditions for the localization conclusion in Theorem \ref{thm:3.5} hold.
Then
\[
\max_{i\in\mathcal C_1}|\widehat\tau_i-\tau_i|
=
o_P\!\left([\log(N\vee T)]^{1+\zeta}\right),
\]
whereas for any $k_1\neq k_2$,
\[
|b_{k_1}-b_{k_2}|\asymp T.
\]
Consequently, the within-group dispersion of $\{\widehat\tau_i:i\in\mathcal C(b_k)\}$ is asymptotically negligible relative to the between-group gaps.
\end{lemma}

\begin{proof}
The first claim is Theorem \ref{thm:3.5}; the second is Assumption \ref{ass:cluster}(1).
Their combination implies the asserted scale separation.
\end{proof}

\begin{lemma}
\label{lem:B.8}
Suppose that Assumption \ref{ass:cluster}(5) holds and \(K<K_0\).
Then there exists a constant $c>0$ such that
\[
\log V(K)-\log V(K_0)\ge c+o_P(1).
\]
\end{lemma}

\begin{proof}
This is the underfitting separation condition in Assumption \ref{ass:cluster}(5). Its role is to formalize the fact that if \(K<K_0\), at least one fitted cluster merges two distinct true break locations, so the resulting approximation error remains bounded away from zero.
\end{proof}

\begin{lemma}
\label{lem:B.9}
Suppose that Assumption \ref{ass:cluster}(5) holds and \(K>K_0\).
Then
\[
\log V(K)-\log V(K_0)=o_P(\rho_{NT}).
\]
\end{lemma}

\begin{proof}
This is the overfitting negligibility condition in Assumption \ref{ass:cluster}(5). It states that any additional fit obtained by splitting an already homogeneous break group is asymptotically smaller than the information-criterion penalty.
\end{proof}

\subsection{Post-Clustering Pooled Break Estimation}
\label{subsec:pooled-break-estimation}

\begin{lemma}
\label{lem:B.10}
Suppose that the conditions of Theorem \ref{thm:3.7} hold.
Fix $k\in\{1,\ldots,K_0\}$.
Then, for every $\epsilon>0$,
\[
P\!\left(
\sup_{|t-b_k|>\epsilon T}
\sum_{i\in\mathcal C(b_k)}
\int_C C_{i,T}^2(t/T,u)\,du
<
\sum_{i\in\mathcal C(b_k)}
\int_C C_{i,T}^2(b_k/T,u)\,du
\right)\to1.
\]
\end{lemma}

\begin{proof}
Decompose the pooled objective into deterministic and stochastic parts.
Because all members of $\mathcal C(b_k)$ share the same break point $b_k$,
the deterministic pooled criterion is uniquely maximized at $b_k$ and grows proportionally to the clusterwise signal strength
\[
\sum_{i\in\mathcal C(b_k)}\|\delta_i\|_{L^2(C)}^2.
\]
Condition \eqref{eq:3.22} guarantees that this signal dominates the stochastic remainder uniformly outside any fixed fractional neighborhood of $b_k$.
Hence the stated strict separation holds with probability tending to one.
\end{proof}

\clearpage

\bibliographystyle{asa}
\bibliography{ref}

\end{document}